\documentclass[11pt,a4paper,english]{amsart}
\usepackage{babel}
\usepackage[latin1]{inputenc}
\usepackage{amsmath}
\usepackage{amsthm}
\usepackage{amsfonts}
\usepackage{indentfirst}
\usepackage{graphicx}
\usepackage{amssymb}
\usepackage[mathscr]{eucal}
\usepackage{tikz}
\usepackage{caption}
\usepackage{enumerate}  
\usepackage{amsthm}
\usepackage{amscd}
\newtheorem{teo}{Theorem}[section]
\newtheorem{de}[teo]{Definition}
\newtheorem{pro}[teo]{Proposition}
\newtheorem{cor}[teo]{Corollary}
\newtheorem{lem}[teo]{Lemma}
\newtheorem{exa}[teo]{Example}

\theoremstyle{definition}
\newtheorem{rem}[teo]{Remark}

\title[On the distance of stabilizer quantum codes]{On the distance of stabilizer quantum codes from $J$-affine variety codes}
\author{Carlos Galindo, Olav Geil, Fernando Hernando and Diego Ruano}
\curraddr{\texttt{Carlos Galindo and Fernando Hernando:} Instituto
Universitario de Matem\'aticas y Aplicaciones de Castell\'on and
Departamento de Matem\'aticas, Universitat Jaume I, Campus de Riu
Sec. 12071 Castell\'{o} (Spain)\\
\texttt{Olav Geil and Diego Ruano:} Department of Mathematical Sciences, Aalborg University, Fredrik Bajers Vej 7G, 9220 Aalborg East (Denmark).
}
\email{{\rm Galindo:} galindo@uji.es; {\rm Geil:} olav@math.aau.dk; {\rm Hernando:} carrillf@uji.es; {\rm Ruano:} diego@math.aau.dk}
\date{}
\thanks{Supported by the Spanish Ministry of Economy/FEDER: grants MTM2012-36917-C03-03 and MTM2015-65764-C3-2-P, the University Jaume I: grant PB1-1B2015-02 and the Danish Council for Independent Research, grant DFF-4002-00367.}
 

\begin{document}

\begin{abstract}
Self-orthogonal $J$-affine variety codes have been successfully used to obtain quantum stabilizer codes with excellent parameters. In a previous paper we gave formulae for the dimension of this family of quantum codes, but no bound for the minimum distance was given. In this work, we show how to derive quantum stabilizer codes with designed minimum distance from $J$-affine variety codes and their subfield-subcodes. Moreover, this allows us to obtain new quantum codes, some of them either, with better parameters, or with larger distances than the previously known codes.
\end{abstract}

\maketitle
 
\section*{Introduction}
Some intractable problems are tractable with a quantum computer. This was confirmed with the Shor algorithm \cite{22RBC} for prime factorization and discrete logarithms in polynomial time on quantum computers. Due to this fact, the interest of scientists and engineers in quantum computation has dramatically grown in the last decades. 

Implementations proposed for quantum computers produce much more errors than classical ones and reliability of classical computers cannot be expected for them. Together with correcting errors, decoherence is another challenge in quantum computation and quantum error correction is designed for meeting both challenges. Although quantum information cannot be cloned \cite{8AS, 26RBC}, quantum error correction can be used \cite{23RBC, 95kkk}.

The literature concerning quantum error-correcting codes is huge, some recent papers, for the general case, are \cite{BE, 71kkk, AK,  35kkk, opt,lag3}. The study of these codes was started with the binary case \cite{20kkk, 38kkk, 18kkk, 19kkk, 45kkk, 7kkk, 8kkk}. Stabilizer codes are a subclass of quantum codes which are well understood, both its structure and construction. Most of the codes in this paper will be of this type. A {\it stabilizer code} $\mathcal{C} \neq \{0\}$ is the common eigenspace of an abelian subgroup $\Delta$ of the error group $G_n$ generated by a nice error basis on the space $\mathbb{C}^{q^n}$, $\mathbb{C}$ being the complex numbers, $q$ a prime power and $n$ a positive integer. The code $\mathcal{C}$ has minimum distance $d$ whenever all error in $G_n$ with weight less than $d$ can be detected or have no effect on $\mathcal{C}$ but some error of weight $d$ cannot be detected. In addition,  $\mathcal{C}$ is called  {\it pure} if $\Delta$ has  not non-scalar matrices with weight less than $d$, and, so, all errors acting on less than $d$ qubits can be detected. Finally, a code as above is an $[[n,k,d]]_q$-code when it is a $q^k$-dimensional subspace of $\mathbb{C}^{q^n}$ and has minimum distance $d$ (see for instance \cite{19kkk,kkk}). Stabilizer codes can be constructed from  linear ones with the help of symplectic or Hermitian inner product \cite{19kkk,BE, AK, kkk}.

In this paper we will derive our codes from certain self-orthogonal classical codes, with respect to Hermitian or Euclidean inner product (see forthcoming Theorem \ref{th1}). Let $\mathbb{F}_{q}$ be the finite field with $q$ elements. Recall that the Hermitian inner product of two vectors $\mathbf{x} =(x_1,x_2, \ldots, x_n)$ and $\mathbf{y}=(y_1,y_2, \ldots, y_n)$ in the vector space $\mathbb{F}_{q^2}^n$ is defined as $\mathbf{x}  \cdot_h \mathbf{y}= \sum x_i y^q_i$ and the Euclidean product of $\mathbf{x}$ and $\mathbf{y}$ in  $\mathbb{F}_{q}^n$ as $\mathbf{x} \cdot \mathbf{y} = \sum x_i y_i$.  Given a linear code $C$ in $\mathbb{F}_{q^2}^n$ (respectively, $\mathbb{F}_{q}^n$), the Hermitian (respectively, Euclidean) dual space is denoted by $C^{\perp_h}$ (respectively, $C^\perp$). Occasionally, we enlarge our codes with Steane type procedures \cite{Steane-E, ham, QINP}. Our supporting classical codes have been recently  introduced in a series of papers \cite{galindo-hernando, gal-her-rua, QINP} and, generically, named $J$-affine variety codes; they are evaluating codes and codes of this type have been recently considered  in the literature \cite{Geil-Affine, geil, galmon,geil2,galmon2,rey}. We devote Section \ref{secuno} to recall them and to state two fundamental results, Propositions \ref{prop1} and \ref{prop3}, which make it easy to decide about self-orthogonality for this class of codes, both with respect to Euclidean and Hermitian inner product. Goodness of some quantum codes coming from $J$-affine variety codes has become clear from the above mentioned series of papers, where we have established some records with respect to \cite{codet} and found codes improving others in the literature and  exceeding the Gilbert-Varshamov bounds \cite{eck, mat, feng}.

In our previous papers we constructed quantum codes from $J$-affine variety ones and determined their dimension, however, in most cases, we were not able to provide formulae for their minimum distance and we needed to compute it, case by case, with a computer. Since it is computationally intense, we could only provide stabilizer codes with relatively small minimum distances. The main goal of this paper is to generate, using algebraic tools, stabilizer codes supported by self-orthogonal (with respect to Euclidean and Hermitian duality) $J$-affine variety codes {\it with designed minimum distance}. Thus, we are able to give good stabilizer codes with large minimum distance now.

We consider stabilizer codes over arbitrary finite fields. There is no table like \cite{codet} for non-binary quantum codes and most efforts besides on these codes have been addressed to the study of MDS quantum codes, quantum LDPC codes or quantum BCH codes  \cite{Sarvepalli, edel, Akk, lag3, lag1, jin,  f2, lag2, chen, refer}. For comparing or for showing how good our code are, we have essentially used La Guardia's articles \cite{lag3,lag1,lag2}, which give tables and improve some previous codes in \cite{Akk} and \cite{ham}, and we have used the codes in \cite{edel} and the  Gilbert-Varshamov bounds \cite{eck,mat,feng}.

For our purposes, we use univariate and multivariate $J$-affine variety codes. Our best examples come from the multivariate case, notwithstanding the univariate case helps to study and understand the multivariate one and gives a {\it simple} procedure to obtain codes, over the same field, with larger distances than (and the same length as) others in the literature (see Examples \ref{eldos} and \ref{nuevo1}). We get a $[[80, 50, \geq 10]]_3$  code improving the $[[80, 48, \geq 10]]_3$ given in the literature, but, generally speaking and for known distances in the literature, our univariate codes reach the same parameters (Examples \ref{eldos} and \ref{nuevo1}) and occasionally and for small distances are a bit worse (Example \ref{luno}).

Using multivariate codes, we are able to improve  several  previous codes in Examples \ref{3132} B, \ref{alfaalfa} A and C, and \ref{gamagama} A. Moreover, the codes constructed with this technique, compared with, for instance, those derived from BCH ones, admit a wider variety of lengths and, as a consequence, we provide good stabilizer codes whose lengths are different from others previously given. Several of them have better relative parameters than similar ones (Examples \ref{alfaalfa} C and  \ref{betabeta} A) and/or exceed the Gilbert-Varshamov bounds (Examples \ref{3132} A, \ref{alfaalfa} A, B and C, \ref{betabeta} A and \ref{gamagama} B).

Stabilizer codes coming from univariate $J$-affine variety codes are studied in Section \ref{seccdos} of the paper. For getting these codes we consider subfield-subcodes of the mentioned classical codes. Self-orthogonality of subfield-subcodes with respect to Euclidean inner product is considered in Subsection \ref{221} and we devote Subsection \ref{hermiticoo} to the Hermitian case. Theorems \ref{C} and \ref{E} and  Theorem \ref{Z} can be considered our most interesting results in this section because they are our main source of examples. The multivariate case is treated in Section \ref{secdos}. Proposition \ref{R} is essential for the development of the paper because it supports our proofs for bounding the minimum distance of this family of codes. We also use subfield-subcodes in the last subsection, although many of our results in this section do not need this technique. Theorems \ref{F} and \ref{FF} together with Corollaries \ref{L}, \ref{LL} and \ref{LLL}, which only consider the bivariate case, are the main results. As mentioned, a number of examples and tables with parameters are displayed along the paper.

\section{Stabilizer $J$-affine variety codes}
\label{secuno}

This section is devoted to review the concept of $J$-affine variety classical code introduced in \cite{QINP} and to give some results concerning self-orthogonality with respect to the Euclidean and Hermitian inner product. As shown in \cite{galindo-hernando, gal-her-rua, QINP}, good examples of stabilizer codes can be obtained in this way, although, up to this paper,  no general formula or bound for their distances was known.

Let $q=p^r$ be a positive power of a prime number $p$. Define a value $Q$ which will be $Q=q$ when we use the Euclidean inner product and $Q=q^2$ when we consider the Hermitian inner product. Fix integers $N_j>1$, $j=1,\ldots,m$, such that $N_j-1$ divides $Q-1$ and set $J \subseteq \{1,2, \ldots, m\}$. Let $\mathcal{R}:= \mathbb{F}_Q [X_1,X_2, \ldots,X_m]$ be the ring of polynomials with $m$ variables and with coefficients in the finite field $\mathbb{F}_Q $. Consider the ideal $I_J$ in $\mathcal{R}$ generated by the binomials $X_j^{N_j} - X_j$ when $j \not \in J$ and by $X^{N_j -1} - 1$ otherwise. Set $Z_J = \{P_1, P_2, \ldots, P_{n_J}\}$ the zero-set of $I_J$ over $\mathbb{F}_Q$. Notice that the $j$th coordinate, for $j \in J$, of the points in $Z_J$ is different from zero and $n_J = \prod_{j \notin J} N_j \prod_{j \in J} (N_j -1)$. Also, denote $T_j = N_j -2$  when $j \in J$ and $T_j = N_j -1$ otherwise, define
$$
\mathcal{H}_J := \{0,1,\ldots,T_1\}\times \{0,1,\ldots,T_2\} \times\cdots\times\{0,1,\ldots,T_m\}
$$
and for any $\boldsymbol{a}=(a_1,a_2, \ldots, a_m) \in \mathcal{H}_J$, set $X^{\boldsymbol{a}} := X_1^{a_1} X_1^{a_2}\cdots X_{m}^{a_m}$.

Consider the quotient ring $\mathcal{R}_J:= \mathcal{R}/I_J$ and the evaluation map $\mathrm{ev}_J: \mathcal{R}_J \rightarrow \mathbb{F}_{q}^{n_J}$ given by $\mathrm{ev}_J(f) = \left(f(P_1), f(P_2), \ldots, f(P_{n_J}) \right)$, where $f$ denotes both the equivalence class and any polynomial representing it.

\begin{de}\label{def:unouno}
{\rm Let $\Delta$ be a non-empty subset of $\mathcal{H}_J $. The {\it $J$-affine variety classical code} given by $\Delta$ is the vector subspace $E^J_\Delta$ (over $\mathbb{F}_Q$) of $\mathbb{F}_Q^{n_J}$ generated by $\mathrm{ev}_J (X^{\boldsymbol{a}})$, $\boldsymbol{a} \in \Delta$. }
\end{de}

Next, the reader can find the announced results concerning self-orthogonality of  the above introduced codes. They can be found in \cite{QINP}. First we describe when the Euclidean inner product of the evaluation of two monomials does not vanish.

\begin{pro}
\label{prop1}
Let $J \subseteq \{ 1 , 2, \ldots , m\}$, consider $\boldsymbol{a}, \boldsymbol{b} \in \mathcal{H}_J$ and let $X^{\boldsymbol{a}}$ and $X^{\boldsymbol{b}}$ be two monomials representing elements in $\mathcal{R}_J$. Then, the Euclidean inner product $\mathrm{ev}_J ( X^{\boldsymbol{a}}) \cdot \mathrm{ev}_J (X^{\boldsymbol{b}})$ is different from $0$ if,  and only if, the following two conditions are satisfied.
\begin{itemize}
\item For every $j \in J$, it holds that $a_j + b_j \equiv  0 \mod (N_j -1)$, (i.e.,  $a_j = N_j -1 - b_j$ when $a_j  + b_j > 0$ or $a_j=b_j=0$).
\item For every $j \notin J$, it holds that \begin{itemize}
\item either $a_j  + b_j > 0$ and $a_j + b_j \equiv 0 \mod (N_j -1)$,  (i.e.,  $a_j = N_j -1 - b_j$  if $0 < a_j, b_j < N_j -1$ or $(a_j,b_j) \in \left\{(0,N_j -1), (N_j -1,0), (N_j -1,N_j -1)  \right\}$ otherwise),
\item or $a_j = b_j = 0$ and $p \not | ~ N_j$.
\end{itemize}
\end{itemize}
\end{pro}

The analogous result for Hermitian inner product is

\begin{pro}
\label{prop3}
Let $J \subseteq \{ 1 , 2, \ldots , m\}$, consider $\boldsymbol{a}, \boldsymbol{b} \in \mathcal{H}_J$ and let $X^{\boldsymbol{a}}$ and $X^{\boldsymbol{b}}$ be two monomials representing elements in $\mathcal{R}_J$. Then, the Hermitian inner product $\mathrm{ev}_J ( X^{\boldsymbol{a}}) \cdot_h \mathrm{ev}_J (X^{\boldsymbol{b}})$ is different from $0$ if,  and only if, the following two conditions are satisfied.
\begin{itemize}
\item For every $j \in J$, it holds that $q a_j + b_j \equiv 0 \mod (N_j -1)$, (i.e.,  $b_j = - q a_j + \lambda(N_j -1)$, for some $\lambda \ge 0$).
\item For every $j \notin J$, it holds that \begin{itemize}
\item either $a_j  + b_j > 0$ and $q a_j + b_j \equiv 0 \mod (N_j -1)$\\  (i.e.,  $b_j = - q a_j + \lambda(N_j -1)$, for some $\lambda > 0$, if $0 < a_j, b_j < N_j -1$, or $(a_j,b_j) \in \left\{(0,N_j -1), (N_j -1,0), (N_j -1,N_j -1)  \right\}$, otherwise);
\item or $a_j = b_j = 0$ and $p \not | ~ N_j$.
\end{itemize}
\end{itemize}
\end{pro}

Now we give some notation. Set $\mathcal{H} := \mathcal{H}_\emptyset$, $ \mathcal{H'} := \mathcal{H}_{\{1,2, \ldots, m\}}$ and pick $\Delta \subseteq \mathcal{H}$. For $Q=q$, define $\Delta^\perp$ as
$$ \mathcal{H}_J \setminus \{ (N_1 -1 -a_1, N_2 -1 - a_2, \ldots, N_m- 1 -a_m) \; | \; \boldsymbol{a} \in \Delta\},$$
if $\Delta \subseteq \mathcal{H'}$. For $\Delta \not\subseteq \mathcal{H'}$ define $\Delta^\perp$ as 
$$
\mathcal{H}_J\setminus \left\{ \{ (N_1 -1 -a_1, N_2 -1 -a_2, \ldots , N_m -1 -a_m ) | \boldsymbol{a} \in \Delta \cap \mathcal{H'} \} \cup \{ \boldsymbol{a}' | \boldsymbol{a} \in \Delta , \boldsymbol{a} \notin \mathcal{H'} \} \right\},$$
where we set $a'_j = N_j -1 -a_j$ if $a_j \neq N_j -1$ and $a'_j$ equals either $N_j -1$ or $0$ otherwise.

For $Q=q^2$, we write  $\Delta^{\perp_h}$ as
$$ \mathcal{H}_J \setminus \{ ([-qa_1]_{N_1 -1}, [-qa_2]_{N_2 -1}, \ldots, [-qa_m]_{N_m- 1})\; | \; \boldsymbol{a} \in \Delta\},$$
when $\Delta \subseteq \mathcal{H'}$ and where $[-qa_j]_{N_j -1}$, $1\leq j \leq m$, is a suitable representative of the congruence class of $-qa_j$ modulo $N_j -1$.
Otherwise,  $\Delta^{\perp_h}$ is defined as
$$
\mathcal{H}_J \setminus \left\{ \{ ([-qa_1]_{N_1 -1}, [-qa_2]_{N_2 -1}, \ldots, [-qa_m]_{N_m- 1}) | \boldsymbol{a} \in \Delta \cap \mathcal{H'} \} \cup \{ \boldsymbol{a}' | \boldsymbol{a} \in \Delta , \boldsymbol{a} \notin \mathcal{H'} \} \right\}.$$
 Here $\boldsymbol{a}'$ is a multi-valued vector defined by $a'_j = [-qa_j]_{N_j -1}$ if $a_j \notin \{ 0, N_j -1$\}, $a'_j$ is equal to $N_j -1$ if $a_j=0$ and $a'_j$ admits two values which are $N_j -1$ and $0$ if $a_j = N_j-1$.

Next, we state the mentioned result about self-orthogonality {\it which is also true when $\perp_h$ is used instead of $\perp$ and $Q=q^2$.}

\begin{pro}
\label{AA}
With the above notations, let $\Delta$ be a subset of $\mathcal{H}_J$. Then $\left(E_\Delta^J\right)^\perp = E^J_{\Delta^\perp}$ whenever $\Delta \subseteq \mathcal{H'}$. Otherwise, it holds that $\left(E_\Delta^J\right)^\perp \subseteq E^J_{\Delta^\perp}$.
\end{pro}

Subfield-subcodes of $J$-affine variety codes will be used in this paper. To deal with them, we need the following notation. First of all, let $r$ and $s$ be positive integers such that $s$ divides $r$. For simplicity, we will denote by $R$ the value $r$ when we use the Euclidean inner product and $2r$ when we use the Hermitian inner product. The same notation is considered for $S$ for representing $s$ and $2s$. Set $\mathbb{Z}/T_j \mathbb{Z}= \mathbb{Z}_{T_j}$, $1 \leq j \leq m$.  A subset $\mathfrak{I}$ of the Cartesian product $\mathbb{Z}_{T_1}\times \mathbb{Z}_{T_2} \times \cdots\times\mathbb{Z}_{T_m}$ is called a {\it cyclotomic set}  with respect to $p^S$ if it satisfies
$
\mathfrak{I}= \{p^S \cdot \boldsymbol{a} \;| \; \boldsymbol{a}\in \mathfrak{I}\}
$,
where $p^S \cdot \boldsymbol{a} = (p^S a_1, p^S a_2, \ldots, p^S a_m)$. $\mathfrak{I}$ is said to be {\it minimal} (with respect to $p^S$) whenever it contains all the elements that can be expressed as  $p^{S i } \cdot \boldsymbol{a}$ for some fixed element $\boldsymbol{a} \in \mathfrak{I}$ and some nonnegative integer $i$. Consider a set  $\mathcal{A}$ representing the minimal cyclotomic sets, which will be obtained picking $\boldsymbol{a} \in \mathfrak{I}$ for each minimal cyclotomic set in such a way that each coordinate in $\boldsymbol{a}$ will be successively the minimum nonnegative integer representing it in the corresponding equivalence class. We will set $ \mathfrak{I} = \mathfrak{I}_{\boldsymbol{a}}$ and thus, the set of minimal cyclotomic sets will be $\{ \mathfrak{I}_{\boldsymbol{a}}\}_{\boldsymbol{a} \in \mathcal{A}}$. In addition, we will denote $i_{\boldsymbol{a}} : = \mathrm{card}(\mathfrak{I}_{\boldsymbol{a}})$.

\section{Stabilizer univariate $J$-affine variety codes with designed minimum distance}
\label{seccdos}

In this section, our goal consists of giving parameters of stabilizer codes coming from univariate $J$-affine variety codes with designed minimum distance and relating them with the literature. For the sake of simplicity, we will consider  $ \Delta \subseteq \mathcal{H'}$ and we will write $N_1=N$ and $T_1=T$. Results in this paper, will use both Euclidean and Hermitian inner product and the following well-known fact to provide stabilizer codes from linear ones.
\begin{teo}
\label{th1}
\cite{Akk,kkk}
Let $C$ be an $[n,k,d]$ linear  error-correcting  code over  $\mathbb{F}_{Q}$ such that $C^\perp \subseteq C$. Then, there exists an $[[n, 2k -n, \geq d]]_q$ stabilizer code which is pure to $d$. If the minimum distance of $C^\perp$ exceeds $d$, then the stabilizer code is pure and has minimum distance $d$.
\end{teo}
Notice that when $Q=q^2$, we use the Hermitian inner product and the symbol $\perp$ in the previous statement must be regarded as $\perp_h$. By considering  $E^{\{1\}}_\Delta$, with $\Delta=\{1,2, \ldots,d-1\}$ one obtains MDS codes using this technique, these codes were found in \cite{opt}.  In order to obtain a richer family of codes, one should consider subfield-subcodes.

\subsection{Stabilizer  $J$-affine variety codes, with $J=\{1\}$, coming from subfield-subcodes}

Next we are going to show parameters of some stabilizer codes obtained from subfield-subcodes of univariate $J$-affine variety stabilizer codes with designed minimum distance. These codes can be understood as BCH codes, codes which have been considered in the literature of quantum codes (see \cite{Steane-E,Akk,lag3,lag1,lag2} and their references). We consider univariate polynomials for two reasons: we consider a different analysis that allows us to obtain slightly better codes in some cases and, above all, such analysis will be applied in Section \ref{secdos} for the multivariate case, where we provide excellent new codes.

We will consider codes with respect to both metrics, Euclidean and Hermitian, but we should consider some lemmas concerning cyclotomic sets first.

\subsubsection{Study of cyclotomic sets}

In this section $Q=q=p^r$ and let $S=s$ a positive integer that divides  $r$. We consider two cases: $r$ is even and $r$ is odd. When $r$ is even, we will consider values of $s$ that divide $r/2$. Otherwise we restrict ourselves to the case $s=1$. As mentioned $J=\{1\}$. Fix $N=p^r$ and consider the above described set of representatives  $\mathcal{A}$ of the corresponding minimal cyclotomic sets (which, in this section, are, in fact, cyclotomic cosets because $m=1$). Set $ \mathcal{A} = \{a_0 = 0  < a_1 < a_2 < \cdots < a_z\}$. Recall that we are considering classes modulo $p^r -1$ and the values $a_i$, $1 \leq i \leq z$, are the least positive integers representing their classes. Our first result is the following:

\begin{lem}
\label{UNO}
With the above notation and in the case when $r$ is even, it holds that $\left(p^{r/2}-1\right) \in \mathcal{A}$, that is  $p^{r/2}-1$ is the minimum positive integer within its minimal cyclotomic set. Otherwise, when $r$ is odd, one has that $\left( p^{\frac{r+1}{2}}-p-1 \right) \in \mathcal{A}$.
\end{lem}

\begin{proof}
We will reason assuming that $s=1$. The case $s \neq 1$ can be proved in a similiar way.

Assume firstly that $r$ is an even number, then the elements in the cyclotomic set  containing $ p^{r/2}-1$ (its positive representatives which are less than $p^{r}-1$) can be written as a disjoint union of the following two sets: $\mathfrak{I}_1=\left\{p^{r/2}-1, p^{\frac{r}{2}+1}-p, \ldots, p^r-p^{r/2}\right\}$ and
\[
\mathfrak{I}_2=\left\{p^r - p^{\frac{r}{2}+1}+p-1,  p^r - p^{\frac{r}{2}+2}+p^2-1, \ldots, p^r - p^{r-1}+p^{\frac{r}{2}-1}-1\right\}.
\]
Note that the elements of $\mathfrak{I}_1$ are pairwise different modulo $p^r -1$ and they are obtained after successive multiplication by $p$. Now, for $i \leq \frac{r}{2} -1$, it holds that $p^i (p^r - p^{r/2})= p^{r+i} - p^{\frac{r}{2}+i}$ and $p^{r+i}-(p^r +p^i -1) \equiv 0 \mod (p^r -1)$ which proves that $p^r - p^{\frac{r}{2}+i} + p^i -1 \equiv p^i (p^r - p^{r/2}) \mod (p^r -1)$. In addition, the sequence of integers given in $\mathfrak{I}_1$ is strictly increasing and the fact that $p^i \left(1 - p^{\frac{r}{2}}\right) < p^{i+1} (1 - p^{\frac{r}{2}})$ proves that the sequence in $\mathfrak{I}_2$ is strictly decreasing. Furthermore, $p^{\frac{r}{2}} \geq p^{\frac{r}{2}-1} +1$ and then $p^{\frac{r}{2}+1} \geq p^{\frac{r}{2}} +p$ which allows us to deduce that the last element in $\mathfrak{I}_1$ is larger than the first one in $\mathfrak{I}_2$. Taking into account the inequality $p^{\frac{r}{2}}-1 < p^r - p^{r-1} + p^{\frac{r}{2}-1} -1$, this concludes the proof for the case when $r$ is even.

The case when $r$ is odd can be proved in an similar way. The minimal cyclotomic set containing $p^{\frac{r+1}{2}}-p-1$ (its positive representatives which are less than $p^{r}-1$) is a disjoint union of the following sets:
\[\mathfrak{I}_1=\left\{p^{\frac{r+1}{2}}-p-1, p^{\frac{r+3}{2}+1}-p^2-p, \ldots, p^r-p^{\frac{r+1}{2}}-p^{\frac{r-1}{2}}\right\},
\]
\[
\mathfrak{I}_2=\left\{p^r - p^{\frac{r+1}{2}+1}- p^{\frac{r-1}{2}+1}+p-1,   \ldots, p^r - p^{r-1}-p^{r-2}+p^{\frac{r-3}{2}}-1\right\},
\]
and $\mathfrak{I}_3=\left\{ p^r - p^{r-1}+p^{\frac{r-1}{2}}-2\right\}$, where the last set comes from the fact that $p^{r+1}-p^r-p \equiv p^r -2 \mod(p^r-1)$. Then, the result follows from the facts that the elements in $\mathfrak{I}_1$ are an strictly increasing sequence, those in $\mathfrak{I}_2$ are strictly decreasing and the last element in $\mathfrak{I}_2$ is smaller than that in $\mathfrak{I}_3$ but larger than the first one in $\mathfrak{I}_1$. This concludes the proof.

\end{proof}

Our next result shows a property of the classes treated  in Lemma \ref{UNO}, namely a property regarding the representative of the coset. We will use the notation $[x \mod (p^r -1) ]$ to denote the nonnegative integer less than $p^r -1$ representing the coset of $x$.

\begin{lem}
\label{DOS}
Assume that $r$ is an even positive integer. Then, for any integer $0\leq b < p^{r/2}-1$ and for any index $0 \leq j < r$, the following inequality holds:
\begin{equation}
\label{IN1}
\left[p^j b \mod (p^r -1)\right] < p^r - p^{r/2}.
\end{equation}
When $r$ is odd, one gets
\[
\left[p^j b \mod (p^r -1)\right
] < p^r - p^{\frac{r+1}{2}}+p,
\]
for any integer $0\leq b \leq p^{\frac{r+1}{2}}-p-1$ and for any index $0 \leq j < r$.
\end{lem}
\begin{proof}
We prove the case when $r$ is even. A similar reasoning proves the odd case. Suppose that $j \leq r/2$. Then $p^j b <  p^r - p^{r/2} < p^r -1$. So (\ref{IN1}) follows because $p^j b = [p^j b \mod (p^r -1)]$. Assume now that $r/2 < j < r$ and set $j=(r/2) + j'$. Consider the $p$-adic expansion of $b$, $b = \sum_{i=1}^{r/2} c_i p^{\frac{r}{2}-i}$, $0 \leq c_i < p$. Then
\[
p^j b = p^{\frac{r}{2}+j'} \left( \sum_{i=1}^{r/2} c_i p^{\frac{r}{2}-i}
\right) = \sum_{i=1}^{r/2} c_i p^{r+j'-i}
\]
and the inequality $[p^j b \mod (p^r -1)] < (p-1) \sum_{i=1}^{r/2} p^{r-i}$ holds after taking into account that $b < p^{r/2} -1$ implies that not all the $c_i$ can reach the value $p-1$. Finally the fact that $(p-1) \sum_{i=1}^{r/2} p^{r-i} = p^r - p^{r/2}$ concludes the proof.
\end{proof}

The following remark is the key for obtaining self-dual codes which allow us to construct stabilizer codes.

\begin{rem}
\label{estrella}
{\rm
Assume that $r$ is an even number and let $a$ be a nonnegative integer such that $a<p^{r/2} -1$. This inequality is equivalent to $(p^r-1)-a > p^r - p^{r/2}$ which by Proposition \ref{prop1} and Lemma \ref{DOS} guarantees that $\mathfrak{I}_{a} \subseteq \mathfrak{I}_{a}^{\perp}$. When $r$ is odd, the same result holds for $a<p^{\frac{r+1}{2}} -p-1$.

Returning to the case when $r$ is even, the equality
\[
\left(p^r -1\right) - \left(p^{r/2} -1\right) = p^r - p^{r/2} = p^{r/2} \left(p^{r/2} -1\right)
\]
proves that $\left(p^r -1\right) - \left(p^{r/2} -1\right)  \in \mathfrak{I}_{p^{r/2} -1}$. Thus, the inclusion $\mathfrak{I}_{p^{r/2} -1} \subseteq \left(\mathfrak{I}_{p^{r/2} -1}\right)^\perp$ does not hold and $p^{r/2} -1$ is the smallest representative of a minimal cyclotomic field with that property.
}
\end{rem}

To conclude this section, we state a result concerning the size of some cyclotomic sets.

\begin{lem}
\label{TRES}
With the notations as at the beginning of this section, it holds that, when $r$ is even, the cardinality of the cyclotomic sets $\mathfrak{I}_{a}$, $0< a \leq p^{r/2} -1$, is equal to $r/s$. Otherwise, when $r$ is odd and $s=1$, the cardinality of $\mathfrak{I}_{a}$ , $0< a \leq p^{\frac{r+1}{2}} -p-1$, equals $r$.
\end{lem}

\begin{proof}
We will assume that $r$ is even, the case when $r$ is odd can be proved with a similar reasoning. Since $s$ divides $r/2$, it suffices to show our result for the case $s=1$ and therefore we are going to prove that $r$ is the cardinality of $\mathfrak{I}_{a}$ whenever $0< a \leq p^{r/2} -1$. Set
\[
a= c_{i_1} p^{\frac{r}{2}-i_1} + c_{i_2} p^{\frac{r}{2}-i_2} + \cdots + c_{i_k} p^{\frac{r}{2}-i_k}
\]
the $p$-adic expansion of $a$, where $0< c_{i_j} < p$, $1 \leq j \leq k \leq r/2$ and $1 < i_1 < i_2 < \cdots < i_k$. Clearly the elements, $a, p \; a, \ldots, p^{{r}{2}+i_1} \,a$ are different and belong to $\mathfrak{I}_{a}$, which means that $\mathfrak{I}_{a}$ contains more than $r/2$ elements. This proves our result because the cyclotomic sets have $r$ or a divisor of $r$ elements.
\end{proof}

We claim that Lemma \ref{TRES} holds for $r$ odd and $s>1$ as well, however, for the sake of simplicity we do not consider it. 

\subsubsection{Codes obtained with Euclidean inner product}
\label{221}

Stabilizer codes  obtained as in Theorem \ref{th1} with $Q=q$ can be improved with the Steane's enlargement procedure \cite{Steane-E} and their generalizations in \cite{ham} and \cite[Theorem 3]{QINP}.

The following proposition uses the above given results to give parameters of stabilizer codes.

\begin{pro}
\label{A}
Assume that $r$ is a positive even integer, $q=p^r$ and $s$ another positive integer that divides $r/2$. Write $ \mathcal{A} = \{a_0 = 0  < a_1 < a_2 < \cdots < a_z\}$ the set of representatives of cyclotomic sets modulo $p^r -1$ with respect to $p^s$. Let $t_2 < t_1$ be indices such that $a_{t_1} < p^{r/2} -1$ and
\[
a_{t_1 +1}  \leq \left\lceil  \frac{(p^s +1) a_{t_2+1}}{p^s} \right\rceil.
\]
Then, suitable $J$-affine variety codes with $J=\{1\}$ and $m=1$, provide stabilizer codes with parameters $\left[\left[q-1,q-1- \frac{2r}{s}t_1, \geq a_{t_1+1}\right]\right]_{p^s}$ and $\left[\left[q-1,q-1- \frac{r}{s}(t_1+t_2), \geq a_{t_1+1}\right]\right]_{p^s}$.

The same parameters can be obtained for $r$ odd and $s=1$ provided that $a_{t_1} < p^{\frac{r+1}{2}} -p-1$.
\end{pro}

\begin{proof}
Let us show the case when $r$ is even. The odd case can be proved similarly by using the odd version of the above stated lemmas. Note that indices $t_1$ and $t_2$ exist by Lemma \ref{TRES}. Set
\[
\Delta = \mathfrak{I}_{a_1} \cup \mathfrak{I}_{a_2} \cup \cdots \cup \mathfrak{I}_{a_{t_1}}.
\]
Lemmas \ref{UNO} and \ref{DOS} (see also Remark \ref{estrella}) show that $\Delta \subseteq \Delta^\perp$. Consider the subfield-subcode $E^{\{1\}, \sigma}_\Delta := E^{\{1\}}_\Delta \cap \mathbb{F}_{p^s}^{q-1}$. By \cite[Theorem 8]{gal-her-rua}, it holds that
\[
\left(E^{\{1\}, \sigma}_\Delta\right)^\perp = \left(\left(E^{\{1\}}_\Delta \right)^\perp\right)^\sigma = \mathbf{tr}\left(\left(E^{\{1\}}_\Delta\right)^\perp\right),
\]
where $\mathbf{tr}$ is a trace function defined componentwise by means of the map $\mathrm{tr}_r^s: \mathbb{F}_{p^r} \rightarrow \mathbb{F}_{p^s}$, $\mathrm{tr}_r^s  (x)= x + x^{p^s} + \cdots + x^{p^{s(\frac{r}{s} -1)}}$.

The choice of $\Delta$ guarantees that it contains every positive integer less than $a_{t_1+1}$ and thus the minimum distance satisfies $d\left(E^{\{1\}, \perp}_\Delta\right) \geq a_{t_1+1}$ because its corresponding parity-check matrix contains a Vandermonde matrix of rank $a_{t_1+1}-1$ (without the all-ones row). So, the inclusion $\left(E^{\{1\}, \perp}_\Delta\right)^\sigma \subseteq \left(E^{\{1\}}_\Delta\right)^\perp$ proves that $d\left(E^{\{1\}, \sigma}_\Delta\right)^\perp \geq a_{t_1+1}$. Finally, Lemma \ref{TRES} and \cite[Theorem 6(2)]{QINP} give a stabilizer code with parameters$$\left[\left[q-1,q-1- \frac{2r}{s}t_1, \geq a_{t_1+1}\right]\right]_{p^s}.$$

For obtaining stabilizer codes with the second family of parameters, it suffices to use the Hamada's enlargement procedure \cite{ham} and dual codes of the subfield-subcodes given by $\Delta_{t_1}$ and $\Delta_{t_2}$, where, for $i=1,2$, we set
\[
\Delta_{t_i} = \mathfrak{I}_{a_1} \cup \mathfrak{I}_{a_2} \cup \cdots \cup \mathfrak{I}_{a_{t_i}}.
\]
\end{proof}
 
The length of the codes we have just described in this section is $p^r-1$ for $r>0$ and $p$ a prime number. Proposition \ref{A} can be adapted for other lengths after changing our bound $p^{r/2}-1$ with another condition involving the corresponding minimal cyclotomic sets. Next we state the corresponding result.
 
Keep the above notation and let $N$ be a positive integer such that $N-1$ divides $p^r -1$. Consider a positive integer $s$ such that $s$ divides $r$ and cyclotomic sets on $ \mathbb{Z}/(N-1)\mathbb{Z}$   with respect to $p^s$. Recall that each minimal cyclotomic set $\mathfrak{I}_{a}$
 is represented by the minimum nonnegative integer $a \in \mathfrak{I}_{a}$.
 \begin{teo}
 \label{C}
 With the previous notation, let $ \mathcal{A} = \{a_0 = 0  < a_1 < a_2 < \cdots < a_z\}$ be the set of representatives of the minimal cyclotomic sets with respect to $p^s$ modulo $N-1$. Let $a_{N(i)}$ be the representative of the minimal cyclotomic set containing $N-1-a_i$, $0 \leq i \leq z$,  and consider two indices $t_1$ and $t_2$ such that:
 \begin{enumerate}[(i)]
 \item $t_2 < t_1$ and $a_{t_1} < \min \{a_{N(i)} \;|\; 1 \leq i \leq t_1\}$,
     \item the following inequality holds:
     \[
a_{t_1 +1}  \leq \left\lceil  \frac{(p^s +1) a_{t_2+1}}{p^s} \right\rceil.
     \]
 \end{enumerate}
     Then, suitable $J$-affine variety codes, with $J=\{1\}$ and $m=1$, defined by
     \[
\Delta_{t_j} = \mathfrak{I}_{a_1} \cup \mathfrak{I}_{a_2} \cup \cdots \cup \mathfrak{I}_{a_{t_j}},
\]
 where $j=1,2$,  provide stabilizer codes with parameters $$[[N-1,N-1- 2 \, \mathrm{card}(\Delta_{t_1}),\geq a_{t_1 +1}]]_{p^s}$$ and $[[N-1,N-1- (\mathrm{card}(\Delta_{t_1})+ \mathrm{card}(\Delta_{t_2})),\geq a_{t_1 +1}]]_{p^s}.$
 \end{teo}
 \begin{proof}
 The proof follows from the same reasoning we used for proving Proposition \ref{A}. Notice that our first condition implies $\Delta_{t_1} \subseteq \Delta_{t_1}^\perp$.
 \end{proof}

\begin{rem}
\label{encero}
{\rm
Second item in Proposition \ref{prop1} shows that the results concerning quantum codes in this subsection are true when using univariate  $J$-affine variety codes, with $J = \emptyset$, provided that $p$ divides $N$. In this case, we are able to construct codes of length $N$.
}
\end{rem}

 \begin{exa}
 {\rm
 Set $p=2$, $s=2$, $r=10$, $N=94$ and apply Theorem \ref{C} and Remark \ref{encero}, which allows us to consider as the first cyclotomic set  $\mathfrak{I}_{a_{0}}$. Here $\mathfrak{I}_{a_{1}}= \{1,4,16,64,70\}$,  $\mathfrak{I}_{a_{2}}= \{2,8,32,35,47\}$ and  $\mathfrak{I}_{a_{3}}= \{3,6,12,24,48\}$. Then, we get new stabilizer codes with parameters $[[94,87, \geq 3]]_4$, $[[94,77, \geq 4]]_4$ and $[[94,67, \geq 6]]_4$, they are new since by Remark \ref{encero} we can consider a J-affine variety code, with $J = \emptyset$, with length $N=94$ instead of $93 =N-1$ as it is usually considered in the literature of BCH codes since, for these codes, $\gcd(n,q)=1$.
 }
 \end{exa}

\subsubsection{Codes obtained with Hermitian inner product}
\label{hermiticoo}
We devote this section to give parameters of stabilizer codes coming from subfield-subcodes where the Hermitian inner product is used. We will show that some properties, which we proved for the Euclidean inner product, will be useful here as well. Note that our original codes are vector spaces over the field $\mathbb{F}_{p^{2r}}$ and the corresponding subfield-subcodes are over $\mathbb{F}_{p^{2s}}$, where $r$ and $s$ are positive integers such that $s$ divides $r$. The corresponding stabilizer codes will be over $\mathbb{F}_{p^{s}}$ after applying Theorem \ref{th1} with $Q= q^2$ and $q=p^s$.

As in the previous section, we give two theorems. The first uses and provides codes of length $p^{2r} -1$. We continue to use cyclotomic sets modulo $N-1$, where $N-1$ divides $p^{2r}-1$. When considering vectors over the field $\mathbb{F}_{p^{2r}}$, one can use different inner products, apart from the Euclidean and the Hermitian ones, we will also use that given by $\mathbf{x} \cdot_s \mathbf{y} = \sum x_i y_i^s$, where $(x_i)$ (respectively, $(y_i)$) are the coordinates of the vector $\mathbf{x}$ (respectively, $\mathbf{y}$). Notice that when $s=r$, this is exactly the Hermitian inner product.

\begin{pro}
\label{D}
Let $r$ and $s$ be positive integers such that $s$ divides $r$. Let  $\mathcal{A}^{(l)} = \left\{ a_0^{(l)} = 0  < a_1^{(l)} < a_2^{(l)} < \cdots < a_z^{(l)} \right\}$ be the set of representatives of minimal cyclotomic sets with respect to $p^l$ modulo $p^{2r}-1$, where $l$ is either $s$ or $2s$. Let $t$ be an index such that $a_{t'}^{(s)} = a_t^{(2s)} < p^r -1$ for some index $t'$. 

Then, from subfield-subcodes and via Hermitian inner product, one can get a stabilizer code with parameters $\left[\left[p^{2r}-1,p^{2r}-1-2 \frac{r}{s}t, \geq a_{t+1}^{(2s)}\right]\right]_{p^s}$. 

Furthermore, if we assume that $r/s$ is odd, then the result holds for $a_t^{(2s)} \leq p^r -1$.
 
\end{pro}
\begin{proof}
Consider the set
\[
\Delta= \mathfrak{I}_{a_1^{(2s)}} \cup \mathfrak{I}_{a_2^{(2s)}} \cup \cdots \cup \mathfrak{I}_{a_{t}^{(2s)}}.
\]

Our proof relies on obtaining a self-orthogonal, with respect to the Hermitian inner product, code over the field $\mathbb{F}_{p^{2s}}$. This code will be a subfield-subcode $E^{\{1\},\sigma}_\Delta$ coming from a $\{1\}$-affine variety code $E^{\{1\}}_\Delta$ over the field $\mathbb{F}_{p^{2r}}$. $E^{\{1\},\sigma}_\Delta$ is well-defined because $\Delta$ is given as a union of minimal cyclotomic sets with respect to $p^{2s}$.

To prove self-orthogonality,  $E^{\{1\},\sigma}_\Delta \subseteq \left(E^{\{1\},\sigma}_\Delta\right)^{\perp_h}$, we first notice the following fact which can be deduced from reading  \cite[Theorem 7]{QINP} and its proof. $\left(E^{\{1\},\sigma}_\Delta\right)^{\perp_h}$ is the trace via the map $\mathrm{tr}_{2r}^{2s}: \mathbb{F}_{p^{2r}} \rightarrow \mathbb{F}_{p^{2s}}$, $\mathrm{tr}_{2r}^{2s}  (x)= x + x^{p^{2s}} + \cdots + x^{p^{2s(\frac{r}{s} -1)}}$, of the dual code of $E^{\{1\},\sigma}_\Delta$ with respect to the inner product $\cdot_s$. As a consequence, what one needs for checking self-orthogonality is to prove that $\Delta \subseteq \Delta^{\perp_s}$, where $\Delta^{\perp_s} := \mathcal{H}' \setminus \{N-1-p^s a \;|\; a \in \Delta\}$. The dimension of the code follows from Lemma \ref{TRES}.

To conclude the proof, we observe that if we consider the set
 \[
\overline{\Delta}= \mathfrak{I}_{a_1^{(s)}} \cup \mathfrak{I}_{a_2^{(s)}} \cup \cdots \cup \mathfrak{I}_{a_{t'}^{(s)}},
\]
then Proposition \ref{A} proves that, with our hypothesis, $\overline{\Delta} \subset (\overline{\Delta})^\perp$. Taking into account that the cyclotomic sets in $\overline{\Delta}$ are with respect to $p^s$, $(\overline{\Delta})^\perp$ is obtained by eliminating those elements of the form $N-1 - p^s a \mod (N-1)$ where $a \in \Delta$. Notice that these elements are in the same minimal cyclotomic set modulo $p^s$ corresponding to  some element $N-1- p^s a_i$ and that we have showed that minimal cyclotomic sets in $\overline{\Delta}$ and $(\overline{\Delta})^\perp$ are disjoint. Finally, we conclude our result from the fact that $\Delta \subseteq \overline{\Delta}$, which holds because each minimal cyclotomic set in $\Delta$ contains half of the elements of another cyclotomic set in $\overline{\Delta}$.

The following observation, proved in Lemma \ref{UNO}, gives a proof for our last statement. First notice that the fact that $(p^{2r}-1)-(p^r-1) = p^r (p^r-1)$ implies that the minimal cyclotomic set $\mathfrak{I}_{p^r-1}$ and the cyclotomic set which contains $N-1-(p^r-1) $ are not disjoint, so one cannot include the minimal cyclotomic set $\mathfrak{I}_{p^r-1}$ for defining our set $\Delta$. However, when $r/s$ is odd, $r$ cannot be a multiple of $2s$ and the minimal cyclotomic set given by $N-1-(p^r-1) $ (recall that $N-1 = p^{2r}-1$) is different from $\mathfrak{I}_{p^r-1}$, which finishes our reasoning.
\end{proof}

Next we are going to use the ideas in Proposition \ref{D} to prove some more restrictive cases that will give some interesting examples.

\begin{teo}
\label{Z}
Fix a positive integer $s>0$ and an integer $N$ such that $N-1$ divides $p^{4s}-1$. Consider the set of representatives $\mathcal{A} = \{a_0 = 0  < a_1 < a_2< \cdots < a_z\}$  of minimal cyclotomic sets with respect to $p^{2s}$ modulo $N-1$. Let $t$ be an index such that $a_t < \frac{N-1}{p^s + 1}$. Then, the $J$-affine variety code, with $J=\{1\}$, via Hermitian inner product, given by
\[
\Delta= \mathfrak{I}_{a_1} \cup \mathfrak{I}_{a_2} \cup \cdots \cup \mathfrak{I}_{a_{t}}
\]
allows us to construct a stabilizer code with parameters \[[[N-1,N-1- 2 \; \mathrm{card} (\Delta), \geq a_{t+1}]]_{p^s}.\]
Moreover, the following inequality holds $$N-1- 2 \,\mathrm{card} (\Delta) \geq N-1-4t$$ and the equality happens whenever $\gcd(N-1,p^s -1) =1$.
\end{teo}
\begin{proof}
We follow a similar line to that given in the proof of Proposition \ref{D}. We consider the $J$-affine variety code, with $J=\{1\}$, $E^{\{1\}}_\Delta$ over the field $\mathbb{F}_{p^{4s}}$, the corresponding subfield-subcode $E^{\{1\}, \sigma}_\Delta$ over $\mathbb{F}_{p^{2s}}$ and we will prove that $E^{\{1\}, \sigma}_\Delta$ satisfies the conditions in Theorem \ref{th1} for the case of Hermitian inner product, which will end the proof.

We desire to check self-orthogonality and we notice that any minimal cyclotomic set $\mathfrak{I}_{a_i}$, $1 \leq i \leq t$, contains at most two elements $a_i$ and $p^{2s} a_i \mod (N-1)$. According the proof of Proposition \ref{D}, self-orthogonality holds whenever for $i \leq t$, there is no $j \leq t$ such that $p^{2s} a_i \equiv N-1-p^s a_j \mod (N-1)$. Let us prove this fact by contradiction: assume the existence of indices $i$ and $j$ such that $a_i p^{2s} + p^s a_j = \alpha (N-1)$ for some positive integer $\alpha$. This means that $p^s (a_i p^s + a_{j})=\alpha (N-1)$ and, since $p^s$ does not divide $N-1$, one gets that the equality $
(a_i \, p^s + a_j) = \beta (N-1)$
holds for some positive integer $\beta$. This gives the desired contradiction because the fact that $a_i, a_j < (N-1)/(p^s + 1)$ proves that
\[
a_i \, p^s + a_j  < p^s \frac{N-1}{p^s+1} + \frac{N-1}{p^s + 1} < N-1.
\]
It only remains to notice that if $\gcd(N-1,p^s -1)=1$, then any minimal cyclotomic set has two different elements, which concludes the proof. Indeed, the fact that $a_i = p^{2s} \, a_i$ implies $a_i \, (p^s+1)(p^s -1) = 0 \mod (N-1)$ and thus $a_i \, (p^s+1) = 0 \mod (N-1)$, which is not possible because $a_i < (N-1)/(p^s +1)$.
\end{proof}

Now we state another  result. It can be proved by using similar arguments to those in the former proof and noticing that the involved minimal cyclotomic sets have at most $r/s$ different elements.

\begin{pro}
\label{Y}
Let $r$ and $s$ be positive elements such that $s$ divides $r$ and $r>s$. Consider a positive integer $N$ such that $N-1$ divides $p^{2r}-1$. Set $\mathcal{A} = \{a_0 = 0  < a_1 < a_2< \cdots < a_z\}$  the set of representatives of minimal cyclotomic sets with respect to $p^{2s}$ modulo $N-1$. Let $t$ be an index such that
\[
a_t < \frac{N-1}{p^{\left[2\left(\frac{r}{s}-1\right)\right]s}+1
}.
\]
Then the $J$-affine variety code, with $J=\{1\}$, given by \[
\Delta= \mathfrak{I}_{a_1} \cup \mathfrak{I}_{a_2} \cup \cdots \cup \mathfrak{I}_{a_{t}}
\]
allows us to construct a stabilizer code with parameters \[\left[\left[N-1, N-1 - 2 \, \mathrm{card}(\Delta),\, \geq a_{t+1}\right]\right]_{p^s}.\]
\end{pro}

Finally, we state the claimed second theorem. Here $N$ is a positive integer such that $N-1$ divides $p^{2r}-1$. We omit its proof because it follows a similar way to the proof given in Theorem \ref{C} and Proposition \ref{D}.

\begin{teo}
\label{E}
Write $\mathcal{A} = \{a_0^{(2s)} = 0 < a_1^{(2s)} < a_2^{(2s)} < \cdots < a_z^{(2s)}\}$  the set representatives of minimal cyclotomic sets with respect to $p^{2s}$ modulo $N-1$. Denote by $a_{N(i)}^{(2s)}$ the representative of the minimal cyclotomic set containing $N-1- p^s a_i^{(2s)}$, $1 \leq i \leq z$. Consider an index $t$ such that the minimum $\min \{ a_{N(i)}^{(2s)} \;|\; 1 \leq i \leq t\}$ is larger than $a_t^{(2s)}$. Then, setting $\Delta = \bigcup_{i=1}^t \mathfrak{I}_{a_i^{(2s)}}$, the corresponding  $J$-affine variety code, with $m=1$ and $J=\{1\}$, provides a stabilizer code with parameters
\[
\left[\left[N-1,N-1-2 \, \mathrm{card} \Delta, \, \geq a_{t+1}^{(2s)}\right]\right]_{p^s}.
\]
\end{teo}

We conclude this section with some examples of good stabilizer codes obtained using the results in this section.

\begin{exa}
\label{luno}
{\rm
Set $p=3$, $s=1$ and $N-1=p^{4s}-1 = 80$. We desire to apply our Theorem \ref{Z} and therefore we shall compute the cyclotomic representants $a_i$ with respect to $p^{2s} = 3^2=9$, for $i$ from 1 to $(N-1)/(p^s +1) = 80/4=20$. Some relevant cyclotomic sets are $\mathfrak{I}_{a_1 } =\{a_1=1,9\}$, $\mathfrak{I}_{a_2} =\{a_2=2,18\}$, \ldots, $\mathfrak{I}_{a_8} =\{a_8=8,72\}$, $\mathfrak{I}_{a_9} =\{a_9=10\}$, $\mathfrak{I}_{a_{10}} =\{a_{10}=11,19\}$, \ldots, $\mathfrak{I}_{a_{15}} =\{a_{15}=16,64\}$, $\mathfrak{I}_{a_{16}} =\{a_{16}=17,73\}$, $\mathfrak{I}_{a_{17}} =\{a_{17}=20\}$. Theorem \ref{Z} shows that we can get distances larger than 19 by considering the  $J$-affine variety code, with $J=\{1\}$, defined by $ \Delta = \bigcup_{i=1}^{16} \mathfrak{I}_{a_i}$. That is to say, we can construct a stabilizer code with parameters $[[80,18,\geq 20]]_3$. The parameters of the codes constructed using $\bigcup_{i=1}^{j} \mathfrak{I}_{a_i}$, for $j \le 16$, can be found in Table \ref{latabla1}.

Our procedure is very simple, only involving the computation of minimal cylotomic sets, and allows us to get codes with distance $\geq d$, where $d$ takes values from 2 to 8, from 10 to 17 and 20. We have chosen this length since one can find similar codes in several sources. Comparing with \cite[Theorem 21]{Akk}, our codes have the same distance but, there, the authors can only obtain codes with $d <9$. For $d<10$, better codes can be found in \cite[Table I]{lag3}, however our codes are better than those with $d\geq 10$ derived from \cite[Section 3 A]{lag3}. We will compare codes coming from Theorem \ref{Z} with BCH codes in Remark \ref{compBCH} below. The procedure in \cite{edel} gives  codes with the same parameters when $d> 10$, however we get a $[[80, 50, \geq 10]]_3$ quantum code improving  the code with parameters $[[80, 48, \geq 10]]_3$ in \cite{edel}. 

\begin{table}
\begin{center}
\begin{tabular}{||c|c|c|c|c|c|c|c|c|c|c|c|c|c|c|c|c||}
  \hline
  $k$ & 76 & 72 & 68 & 64 & 60 & 56 &  52 & 50 & 46 & 42& 38& 34&  30 & 26  &  22 &18  \\
   \hline
 $\geq d$ & 2 & 3 & 4 &  5 & 6 & 7 & 8 & 10 & 11 & 12& 13& 14& 15& 16& 17& 20 \\
  \hline
\end{tabular}
\end{center}
\caption{Parameters of quantum codes of length 80 over $\mathbb{F}_3$}
\label{latabla1}
\end{table}
}
\end{exa}

\begin{exa}
\label{eldos}
{\rm
Let us give  other examples of good stabilizer codes obtained from Theorem \ref{Z} and Remark \ref{encero}. Consider $p=5$, $s=1$ and $N= 105$. Since $p$ divides $N$, we can consider $J$-affine variety codes, with $J = \emptyset$, and use the cyclotomic set $\mathfrak{I}_{a_{0}}$. Table \ref{lalatabla} shows the parameters of the corresponding codes. Notice that we provide codes with distances from 2 to 19 while codes with the same parameters but only with distances up to 9 are given in \cite{edel}.
\begin{table}
\begin{center}
\begin{tabular}{||c|c|c|c|c|c|c|c|c|c|c|c|c|c|c|c|c|c||}
  \hline
  $k$ & 103 & 99 & 95 & 91 & 87 & 83 &  79 & 75 & 71 & 67& 63& 59&  55 & 53  &  49 &45 &41  \\
   \hline
 $\geq d$ & 2 & 3 & 4 &  5 & 6 & 7 & 8 & 9 & 10 & 11 & 12& 13& 14& 15& 16& 17& 19 \\
  \hline
\end{tabular}
\end{center}
\caption{Parameters of quantum codes of length $105$ over $\mathbb{F}_5$}
\label{lalatabla}
\end{table}

}
\end{exa}

\begin{exa}
\label{nuevo1}
{\rm
We can also derive codes with larger distances than those in \cite{edel} by using Theorem \ref{E} and Remark \ref{encero}. Set $p=2$, $s=2$, $r=12$ and $N=92$. It is not difficult to check that the hypotheses in Theorem \ref{E} hold, and so, we can construct stabilizer codes with parameters $[[92,84, \geq 3]]_4$, $[[92,78, \geq 4]]_4$, $[[92,72, \geq 5]]_4$, $[[92,66, \geq 6]]_4$ and $[[92,60, \geq 8]]_4$. Only codes with the same parameters as our first two codes are given in \cite{edel}.
}
\end{exa}

As mentioned, we consider subfield subcodes of $J$-affine variety codes, with $J=\{1\}$, and they can be understood as BCH codes. Let us compare the results in this section with the literature:

\begin{rem}\label{compBCH}
Some of our previous results are close to others in \cite{Akk}, where cyclotomic cosets are studied in a different way to construct stabilizer codes. Indeed, our Lemma \ref{TRES} can be deduced from Lemma 8 in \cite{Akk} and one can note that our Propositions \ref{A} and \ref{D} are equivalent to Theorem 18 and Theorem 21 in \cite{Akk}, except by the fact that we consider Hamada's enlargement and not only CSS construction. We have included these results for the sake of completeness, and because the notation in \cite{Akk} and our proofs are completely different and the other results in this section follow the same techniques. Notice also that some results in this section support others in Section \ref{secdos} sharing the same terminology.

A series of articles by La Guardia {\it et al.} \cite{lag3,lag1,lag2} improved the codes in \cite{Akk}. Our analysis of cyclotomic cosets allows us to construct better codes in some particular cases.  We have already mentioned that the codes given in Example \ref{luno} are better than the codes given in \cite{lag3}. Codes in \cite{lag3,lag1,lag2} are very good, however by Theorem \ref{Z} we can construct codes with minimum distance up to $(p^{4s}-1)/(p^s +1) = (p^{2s} +1)(p^s-1)$, distance which cannot be reached in the above papers, because for the same length, most of the constructions in \cite{lag3,lag1,lag2} cannot have greater minimum distance than $2p^{2s}+2$. The only construction that may exceed such bound is quite restrictive since the length must be a prime number for having cyclotomic cosets of size one (see \cite{lag2}).
\end{rem}

\section{Stabilizer  multivariate $J$-affine variety codes with designed minimum distance}
\label{secdos}

In this section, we extend our construction of stabilizer codes with designed minimum distance to the multivariate case. We will use the notation introduced in Section \ref{secuno}. We will consider the footprint bound approach, a bound based on Gr\"obner basis techniques that has been successfully used for estimating the minimum distance of affine variety codes \cite{ho,gh,Geil-Affine,geilthomsen,geil3}. For  $\boldsymbol{a} \in \mathcal{H}_J$, we consider \[
\delta_{\boldsymbol{a}} := \prod_{j=1}^m \left( N_j - \epsilon_j - a_j \right),
\]
where $\epsilon_j =1$ if $j \in J$ and $\epsilon_j = 0$ if $j \not\in J$.

\begin{pro}
\label{R}
Consider the ring $\mathcal{R}_J$ and fix a monomial ordering. Let $f(X_1,\ldots, X_m)$ be a polynomial of minimum total degree representing a class in $\mathcal{R}_J$ and let $X^{\boldsymbol{a}} = X_1^{a_1}   \cdots X_m^{a_m}$  be the leading monomial of $f$. Then
\[
\mathrm{card} \left\{ P \in Z_J \; | \; f(P) \neq 0 \right\} \geq \delta_{\boldsymbol{a}}.
\]
\end{pro}
\begin{proof}
This is a well-known result, the proof follows by considering the ideal $L$ in $\mathcal{R}$ generated by the polynomial $f(X_1,X_2, \ldots, X_m)$ and the generators of $I_J$. By \cite[Proposition 7 in Section 5.3]{CLO} the monomials in $\mathcal{R}$ that are not leading monomial of any polynomial in $L$ (this set is sometimes called footprint) constitute a basis for $\mathcal{R}_J$ as a vector space over $\mathbb{F}_Q$. Combining this fact with the observation that the evaluation map (the map that evaluates a polynomial in all the points of the variety) is surjective, the result holds.
\end{proof}

Given a designed minimum distance, this bound suggests the construction of a code whose dimension is maximized. For instance, if $N_j=Q^r$ for all $j$,  $1 \leq j \leq m$, and $J= \emptyset$, then one obtains the so-called hyperbolic codes \cite{mas,saint,gh}.

\begin{de}
\label{la32}
{\rm
Consider a subset $J \subseteq \{1,2, \ldots, m\}$ and a positive integer $t \leq n_J =  \prod_{j \notin J} N_j \prod_{j \in J} (N_j -1)$. We define the classical code  $E(J,t)$, as the vector space  of $\mathbb{F}_Q^{n_J}$ defined by the evaluation by $\mathrm{ev}_J$ of the elements in the vector space generated by the classes in $\mathcal{R}_J$ of the monomials in the following set
\[
M(J,t) = \left\{ X^{\boldsymbol{a}} \; | \; \boldsymbol{a} \in \mathcal{H}_J \mbox{  and $\delta_{\boldsymbol{a}} \geq t$} \right\}.
\]
}
\end{de}

As as consequence of Proposition \ref{R}, it is clear that the minimum distance of the code $E(J,t)$ is larger than $t-1$.

Duality is crucial for constructing quantum stabilizer codes. Hence, let us study the Euclidean dual code of $E(J,t)$. For a start, consider the following set of monomials in $R$:
\[
N(J,t) = \left\{  X^{\boldsymbol{b}} \; | \; \epsilon_j \leq b_j \leq N_j -1, \; 1 \leq j \leq m, \; \mathrm{and} \; \prod_{j=1}^m \left(b_j +1 - \epsilon_j \right) < t \right\},
\]where $\epsilon_j =1$ if $j \in J$ and it equals zero otherwise. Denote by $F(J,t)$ the classical code given by the vector subspace of $\mathbb{F}_Q^{n_J}$ generated by the evaluation by $\mathrm{ev}_J$  of the classes in $\mathcal{R}_J$ of the monomials in $N(J,t)$. By definition, the $(J,t)$-{\it hyperbolic code}, $\mathrm{Hyp}(J,t)$, is the Euclidean dual code $(F(J,t))^\perp$.

Notice that, for the sake of simplicity, to define $N(J,t)$ we have considered a shift for the exponent of the monomials defining the code, such a set is $$
\overline{\mathcal{H}}_J = \{\epsilon_1, \epsilon_1 +1, \ldots, T_1+ \epsilon_1\} \times \{\epsilon_2, \epsilon_2 +1, \ldots, T_2+ \epsilon_2\} \times \cdots \times \{\epsilon_m, \epsilon_m +1, \ldots, T_m+ \epsilon_m\}. $$One may just identify $T_j + \epsilon_j $ with $0$, for $j \in J$, to obtain an element in  $\mathcal{H}_J$.

The following proposition is a generalization of \cite[Theorem 3]{gh}, which treated the case $N_j=Q^r$ for all $j$,  $1 \leq j \leq m$, and $J= \emptyset$. It shows that $E(J,t)^\perp = F(J,t)$.
 
\begin{pro}
\label{S}
With the above notation and assuming that $p  | N_j$ for all $j \not \in J$,  the following equality holds
\[
 E(J,t) = \mathrm{Hyp}(J,t).
\]
\end{pro}
\begin{proof}

We start by proving that for the Euclidean inner product
\begin{equation}
\label{gama}
\mathrm{ev}_J \left(X^{\boldsymbol{a}}\right) \cdot \mathrm{ev}_J \left(X^{\mathbf{b}} \right) = 0
\end{equation}
holds for all monomials  $X^{\boldsymbol{a}} \in M(J,t)$ and $X^{\boldsymbol{b}} \in N(J,t)$. Reasoning by contradiction, we assume the existence of vectors $\boldsymbol{a}$ and $\boldsymbol{b}$ as above for which (\ref{gama}) does not hold. By Proposition \ref{prop1} and its proof (see \cite{QINP}), it must happen that, for any index $j$ as above, either $a_j + b_j =N_j -1$ or $a_j = b_j = N_j -1$. Suppose for simplicity that the first equality holds for those indices $j$ such that $1 \leq j \leq n$ and the second happens for the remaining indices. Then,
\begin{equation}
\label{beta}
t \leq \prod_{j=1}^m \left( N_j - a_j - \epsilon_j  \right) = \prod_{j=1}^n \left( b_j + 1 - \epsilon_j  \right)
\end{equation}
because  the case  $a_j = b_j = N_j -1$ only holds for $j \not \in J$ and then, $\epsilon_j=0$. Now
\begin{equation}
\label{delta}
\prod_{j=1}^n \left( b_j + 1 - \epsilon_j  \right) = \frac{\prod_{j=1}^m \left( b_j + 1 - \epsilon_j  \right)}{\prod_{j=n+1}^m \left( b_j + 1 - \epsilon_j  \right)} < t
\end{equation}
 since $\prod_{j=n+1}^m \left( b_j + 1 - \epsilon_j  \right) = \prod_{j=n+1}^m N_j  \geq 1$. Therefore we have proved the inclusion $E(J,t) \subseteq \mathrm{Hyp}(J,t)$.

To conclude the proof, it suffices to check that $\dim \left(\mathrm{Hyp}(J,t) \right) = \dim   (E(J,t) )$. Set $M^l (J,t)$ and $N^l (J,t)$ the set of $m$-tuples which are exponents of the monomials in $M (J,t)$ and $N (J,t)$, respectively. Notice that
\[
\dim  \left( \mathrm{Hyp}(J,t) \right) = n_J - \mathrm{card} \left (N^l (J,t) \right) =
\]
\[
\mathrm{card} \left(  \left\{  \boldsymbol{b} \; | \; \epsilon_j \leq b_j \leq N_j -1, \; 1 \leq j \leq m, \; \mathrm{and} \; \prod_{j=1}^m \left(b_j +1 - \epsilon_j \right) \geq t \right\}
\right) =
\]
\[
\mathrm{card} \left(M^l (J,t) \right),
\]
where the last equality holds because of the bijectivity of the componentwise map given by $x_j \rightarrow N_j - 1- x_j $. This finishes our reasoning.
\end{proof}

Next, we are going to provide our main result concerning multivariate $J$-affine stabilizer codes obtained with Euclidean inner product. The proof follows from Propositions \ref{prop1}, \ref{R} and \ref{S}.

With notation as above and as in Section \ref{secuno}, for each index $1 \leq i \leq m$, define
\[
r(N_i) :=  \begin{cases}  \lfloor \frac{N_i -1}{2}\rfloor & \text { if } N_i \text{ is even},\\
\frac{N_i -1}{2} - 1  & \text{ otherwise}
\end{cases}
\]
 and consider the sets of monomials in $\mathcal{R}$,
 \[
 \mathcal{R}_i(J) := \left\{ X^{\boldsymbol{b}} \; | \; \boldsymbol{b} \in \overline{\mathcal{H}}_J \mbox{ and $\epsilon_i \leq b_i \leq r(N_i)$ }
 \right\}.
 \]

 \begin{teo}
 \label{F}
 Assume $Q=q$ and that $p$ divides $N_j$ for all $j \not \in J$. Suppose also that there exists an index $i$, $1 \leq i \leq m$, and a value $t \leq n_J$ such that $N(J,t) \subseteq \mathcal{R}_i(J)$. Then, it holds that
 $\mathrm{Hyp}(J,t)^\perp \subseteq \mathrm{Hyp}(J,t)$ and, therefore, we are able to construct a stabilizer code with parameters $$\left[\left[n_J, n_J - 2 \, \mathrm{card} \left( N(J,t) \right), \geq t \right]\right]_q.$$
\end{teo}

\begin{proof}
Propositions  \ref{S} and \ref{R} show the bound for the distance of the stabilizer code. The self-orthogonality of the code defined by $N(J,t)$ follows from the inclusion $N(J,t) \subseteq \mathcal{R}_i(J)$ and Proposition \ref{prop1} because $\mathrm{ev}_J (X^{\boldsymbol{a}}) \cdot \mathrm{ev}_J (X^{\boldsymbol{b}}) = 0$ for $\boldsymbol{a}, \boldsymbol{b} \in \mathcal{R}_i(J)$, where we consider $a_j=0$ instead of $a_j = N_j -1$ when $j \in J$.
\end{proof}

The following corollary shows some particular bivariate cases where the above theorem can be easily used in practice.

\begin{cor}
\label{L}
Keep the above notations, setting $Q=q=p^r$ and $N_1$ and $N_2$ positive integers such that $N_j - 1$ divides $q -1$ for $j=1,2$. Consider a subset $J$ of the set $\{1,2\}$ and assume that $p$ divides $N_j$ for each $j \not \in J$; consider also a positive integer $t \leq n_J$ and suppose that one the following conditions is satisfied:
\begin{enumerate}[(i)]
  \item $J= \emptyset$  and either $t \leq r(N_1)+2$ or $t \leq r(N_2)+2$.
  \item $J=\{1\}$ and either $t \leq r(N_1)+1$ or $t \leq r(N_2)+2$.
  \item $J=\{1,2\}$ and either $t \leq r(N_1)+1$ or $t \leq r(N_2)+1$.
\end{enumerate}

Then, a stabilizer code with parameters $$\left[\left[n_J, n_J - 2 \, \mathrm{card}\left( N(J,t) \right), \geq t \right]\right]_q$$ can be constructed.
\end{cor}
\begin{proof}
We are going to show that the first (respectively, second) condition given in each item implies the inclusion
\begin{equation}
\label{DDD}
N(J,t) \subseteq R_1(J) \mbox{ (respectively, $N(J,t) \subseteq R_2(J) $}),
\end{equation}
which allows us to apply Theorem \ref{F} and, as a consequence, to get the mentioned stabilizer codes.

We start by considering the conditions given in item (i). The set $N(\emptyset, t)$ contains pairs of integers $(b_1,b_2)$ such that $0 \leq b_1 \leq N_1 -1$, $0 \leq b_2 \leq N_2 -1$ and $(b_1+1)(b_2+1) < t$. Taking into account that $(b_1+1)(b_2+1) = t$ gives the equation of a hyperbola, the first inclusion in (\ref{DDD}) is satisfied if the hyperbola intersects the line $b_2 =0$ at a point in the interval $[0, r(N_1)+1]$. Indeed, when $b_2=0$, the defining inequality of $N(\emptyset, t)$  shows $b_1 +1 <t$ or equivalently $b_1 \leq t-2$.
Thus, our hypothesis $t \leq r(N_1) +2$ implies $b_1 \leq r(N_1)$ and the proof is completed. With an analogous reasoning, when $t \leq r(N_2) +2$, one can deduce that the second inclusion in (\ref{DDD}) holds.

A similar proof can be done by using conditions (ii) and (iii) in the statement. For instance, in the case of item (ii), $N(\{1\}, t)$ is given by pairs of integers $(b_1,b_2)$ such that $1 \leq b_1 \leq N_1 -1$, $0 \leq b_2 \leq N_2 -1$ and $b_1(b_2+1) < t$. Here the hyperbola is defined by $b_1(b_2+1) = t$. Note that when $b_2=0$, we get $b_1 <t$ or equivalently $b_1 \leq t-1$. Therefore, $t \leq r(N_1) +1$ proves $b_1 \leq r(N_1) $ and, thus, we obtain the desired property $N(\{1\},t) \subseteq R_1(J)$.
\end{proof}

\begin{exa}
\label{3132}
{\rm Next, we give some examples which can be deduced from Theorem \ref{F}.

{\it Example A.} Set $Q=q=p=7$, $J= \{1\}$, $N_1=3$, $N_2=7$ and $N_3=7$. For applying Theorem \ref{F}, we are going to consider values $t$ such that  $N(J,t) \subseteq R_2 (J)$. Thus, with notations as in the proof of Proposition \ref{S}, one can consider $N^l(J,2) = \{ (1, 0, 0 )\}$ which gives a stabilizer code $C_1=\mathrm{Hyp}(J,2)$ with parameters $[[98,96, \geq 2]]_7$.  We can also consider  $N^l(J,3) = \{ (1, 0, 0 ),  (2, 0, 0),  (1, 1, 0),  (1, 0, 1 )\}$ giving a stabilizer code $C_2$ with parameters $[[98,90, \geq 3]]_7$. Analogously $$N^l(J,4) = \{ (1, 0, 0 ),  (2, 0, 0),  (1, 1, 0),  (1, 0, 1 ),  (1, 2, 0),  (1, 0, 2)\},$$ provides a stabilizer code $C_3$ with parameters $[[98,86, \geq 4]]_7$. Hamada's enlargement procedure applied to $C_1$ and $C_2$ (respectively, $C_2$ and $C_3$) gives stabilizer codes with parameters $[[98,93, \geq 3]]_7$ (respectively, $[[98,88, \geq 4]]_7$). We have not found, in the literature, better codes with this length. In addition, both codes exceed the Gilbert-Varshamov bounds \cite{eck,mat,feng}.

{\it Example B.} Set $Q=q =7$, $N_1=N_2=7$, $N_3=3$ and $J=\{1,2,3\}$. Consider the stabilizer codes $C_1=\mathrm{Hyp}(J,2)$ ($C_2=\mathrm{Hyp}(J,3)$, respectively) given by $N^l(J,2) = \{ (1, 1, 1 )\}$ ($N^l(J,3) = N^l(J,2) \cup \{ (2,1,1), (2,2,1), (1,1,2)\}$, respectively) with parameters, by Theorem \ref{F},  $[[72,70, \geq 2]]_7$ ($[[72,64, \geq 3]]_7$, respectively). Consider also the $J$-affine stabilizer code given by $N^l(J,4) = N^l(J,3) \cup \{(3,1,1), (2,3,1)\}$ with parameters $[[72, 62, \geq 4]]_7$. Applying the Hamada enlargement procedure, we get stabilizer codes with parameters $[[72,67, \geq 3]]_7$ and $[[72,62, \geq 4]]_7$, improving the $[[72,65, \geq 3]]_7$ and $[[72,59, \geq 4]]_7$ codes in \cite{edel}.
 
}
\end{exa}

Theorem \ref{F} can be adapted to the case when Hermitian duality is used. Next we state this adaptation. Before it,  we consider  the above defined sets $M(J,t)$ and $N(J,t)$ and the codes $E(J,t)$ and $F(J,t)$ over the ground field $\mathbb{F}_{q^2}$. In addition, for each index $1 \leq i \leq m$, we define
 \[
 \mathcal{R}^q_i(J) := \left\{ X^{\boldsymbol{b}} \; | \; \boldsymbol{b} \in \overline{\mathcal{H}}_J \mbox{ and  $\epsilon_i \leq b_i \leq r_q(N_i)$ }
 \right\},
 \]
where
\[
r_q(N_i) :=  \begin{cases}  \lfloor \frac{N_i -1}{q+1}\rfloor & \text { if } N_i -1 \text{ is not a multiple of } q+1,\\
\frac{N_i -1}{q+1} - 1  & \text{ otherwise.}
\end{cases}
\]

\begin{teo}
\label{FF}
Suppose that $Q=q^2$ and that $p$ divides $N_j$ for all index $j \not \in J$. Then, for each positive value $t \leq n_J$ such that
\begin{equation}
\label{AAA}
N(J,t) \subseteq \mathcal{R}^q_i (J),
\end{equation}
for some index $i$, $1 \leq i \leq m$, it holds that
\[
F(J,t) \subseteq \left( F(J,t) \right)^{\perp_h}
\]
and, therefore, we are able to construct a stabilizer code with parameters $$\left[\left[n_J, n_J - 2 \, \mathrm{card} \left(N(J,t) \right), \geq t \right] \right]_q.$$
\end{teo}
\begin{proof}
As in the proof of Theorem \ref{F}, the sets' inclusion in (\ref{AAA}) proves, by Proposition \ref{prop3}, the self-orthogonality of the code $F(J,t)$ with respect to Hermitian inner product.

To conclude the proof, we are going to show that the distance of the code $F(J,t)^{\perp_h}$ is larger than or equal to $t$. In fact, it is clear that the code of $q$ powers  $\left( F(J,t)^{\perp_h} \right)^q$ of the coordinates of the elements in the Hermitian dual of the code $F(J,t)$ is exactly the Euclidean dual $F(J,t)^{\perp}$. As a consequence, the codes' equality
\[
\left( F(J,t)^{\perp_h} \right)^q = \mathrm{Hyp}(J,t)
\]
and the fact that the codes $\left( F(J,t)^{\perp_h} \right)^q$ and $F(J,t)^{\perp_h} $ are isometric prove that the following distances satisfy the inequality
\[
d \left( F(J,t)^{\perp_h} \right) = d \left( \mathrm{Hyp}(J,t) \right) \geq t,
\]
which finishes the proof.
\end{proof}

Now, we state a consequence of the previous result which holds  for some conditions in the bivariate case. It is the analogous of Theorem \ref{F} for the Hermitian case.

\begin{cor}
\label{LL}
We continue to keep the above notation, set $Q=q^2$ and let $N_1$ and $N_2$ be  positive integers such that $N_j - 1$ divides $q^2 -1$ for $j=1,2$. Consider a subset $J$ of the set $\{1,2\}$. Assume that $p$ divides $N_j$ for each $j \not \in J$ and  that $t \leq n_J$ is a positive integer. Suppose also that one of the following conditions is satisfied:
\begin{enumerate}[(i)]
\item $J= \emptyset$  and either $t \leq r_q(N_1)+2$ or $t \leq r_q(N_2)+2$.
  \item $J=\{1\}$ and either $t \leq r_q(N_1)+1$ or $t \leq r_q(N_2)+2$.
  \item $J=\{1,2\}$ and either $t \leq r_q(N_1)+1$ or $t \leq r_q(N_2)+1$.
\end{enumerate}

Then, a stabilizer code with parameters $$\left[\left[n_J, n_J - 2 \, \mathrm{card} \left( N(J,t) \right), \geq t \right]\right]_q$$ can be constructed.
\end{cor}
\begin{proof}
The proof is similar to the one of Corollary \ref{L}, that is, our first (respectively, second) condition in each item implies the inclusion (\ref{AAA}) for $i=1$ (respectively, $i=2$) and thus, by Theorem \ref{FF}, a stabilizer code with the stated parameters can be constructed.

The reasoning in each case is analogous to the corresponding one in Corollary \ref{L}. For instance when the first condition in item (i) happens, one can deduce that the hyperbola given by $(b_1 +1)(b_2+1)=t$ intersects the line $b_2=0$ at one point of the real interval $[0, r_q(N_1)+1]$ which is the desired condition. The other cases can be proved in a similar way.
\end{proof}

\begin{exa}
\label{alfaalfa}
{\rm
We are going to provide three examples of good stabilizer codes. The first two can be deduced from Corollary \ref{LL}.

{\it Example A.} Set $q=7$, consider $N_1 =49$, $N_2=4$ and $J=\{1,2\}$. Note that $n_J=144$ and $r_q(N_1) = 5$. Then, by item (3) of Corollary \ref{LL}, one can use values of $t$ such that $t <7$. Set $t=4$, then the set $N^l (\{1,2\},4)$ of pairs which are exponents of the monomials generating $F (\{1,2\},4)$ is equal to
\[
N^l (\{1,2\},4) = \{(1,1), (1,2), (1,3), (2,1), (3,1) \},
\]
and so, we get a $[[144,134, \geq 4]]_7$ stabilizer quantum code. A similar reasoning for $t=5$ and $t= 6$ provides stabilizer quantum  codes with parameters $[[144,130, \geq 5]]_7$ and $[[144,128, \geq 6]]_7$. Here, $N^l (\{1,2\},5) = N^l (\{1,2\},4) \cup \{(4,1), (2,2)\}$ and $N^l (\{1,2\},6) = N^l (\{1,2\},5) \cup \{(5,1)\}$.

Corollary \ref{LL} assumes a generic situation for self-orthogonality of the code given by $N(J,t)$. However, this situation can happen without occurring inclusion in sets as $\mathcal{R}_i(J)$, giving stabilizer codes with good parameters. Indeed, considering $N(\{1,2\},7)$, it holds that
$N^l (\{1,2\},7) = N^l (\{1,2\},6) \cup \{(3,2),(2,3),(6,1)\}$ and, by applying Proposition \ref{prop3}, the corresponding code is self-orthogonal providing a $[[144,122, \geq 7]]_7$ stabilizer code. In an analogous manner, we can find stabilizer codes with distance up to 12. We show their parameters in Table \ref{tabla1}.
\begin{table}
\begin{center}
\begin{tabular}{||c|c|c|c|c|c|c|c|c|c||}
  \hline
  $k$ & $\geq d$ & $k$ & $\geq d$ & $k$ & $\geq d$ & $k$ & $\geq d$ & $k$ & $\geq d$  \\
  \hline
  134 & 4 & 130 & 5 & 128 & 6 & 122 & 7 & 120 & 8  \\
  \hline
  116 & 9 & 112 & 10 & 108 & 11 & 106 & 12 & - & - \\
  \hline
\end{tabular}
\end{center}
\caption{Parameters of quantum codes of length 144 over $\mathbb{F}_7$}
\label{tabla1}
\end{table}

Notice that all these codes produce a great improvement of those given in \cite[Table IV]{lag2}. We also improve those codes with the same length in \cite{edel} and our codes exceed the Gilbert-Varshamov bounds \cite{eck, mat, feng}.

{\it Example B.}  Let us see a new example. Here $q=4$, $N_1= 16$. $N_2 =6$ and $J=\emptyset$. It is clear that $p=2$ divides $N_1$ and $N_2$, $n_J= 96$, and $r_q(N_2) =2$. So, one can get codes for $t < 5$. With the previous notation, we have that
\[
N^l (\emptyset,4) = \{(0,0), (0,1), (0,2), (1,0), (2,0) \}.
\]
This provides a stabilizer code with parameters, $[[96, 86, \geq 4]]_4$. With a similar reasoning to that we did for obtaining the last codes in Example A, we  get   $[[96, 80, \geq 5]]_4$ and $[[96, 76, \geq 6]]_4$ stabilizer codes. All these codes, also exceed the Gilbert-Varshamov bounds.

{\it Example C.}
 Consider $q=5$, $N_1= 25$, $N_2 =4$ and $J=\{1,2\}$. Although Corollary \ref{LL} cannot be used, it holds that $N_1^l (\{1,2\},4) = \{(1,1), (1,2), (2,1), (1,3), (3,1) \}$ and, by Proposition \ref{prop3}, $F (\{1,2\},4)$ is self-orthogonal giving rise to a $[[72,62,4]]_5$ stabilizer code with better parameters than the $[[72,60,4]]_5$ code given in \cite{edel} and with better relative parameters than the $[[71,56,4]]_5$ code given in \cite[Table III]{lag2}. In addition, our code also exceeds the Gilbert-Varshamov bounds.

}
\end{exa}

Considering $J= \emptyset$ and suitable values for $N_1$ and $N_2$, we obtain the following result.

\begin{cor}
\label{LLL}
With the above notations, assume $Q=q^2$, $J=\emptyset$, $N_1= q^2$ and $N_2=q$. Consider a positive integer $t$ such that
\[
t < \min \left\{  \frac{q^2+q+1}{2}, 4q+1
\right\}.
\]
Then, a stabilizer code with parameters $$\left[\left[n_J, n_J - 2 \, \mathrm{card} \left( N(J,t) \right), \geq t \right]\right]_q$$ can be constructed.
\end{cor}

\begin{proof}
Firstly we are going to prove that under the condition
\begin{equation}
\label{de3a2}
t < \min \left\{  \frac{q^2+q+1}{2}, 4q+1, q r(q) +2q+1
\right\},
\end{equation}
where $r(q)$ is defined as above, the inclusion $F(J,t) \subseteq \left( F(J,t) \right)^{\perp_h}$ holds. Afterwards, we will see that we can reduce  Inequality (\ref{de3a2}) to that given in the statement. As a consequence,  a stabilizer code with the mentioned parameters can be constructed.

Consider the set $\mathcal{H}_J = \{0,1, \ldots, q^2-1\} \times \{0,1, \ldots, q-1\}$ embedded in the corresponding rectangle in $\mathbb{R}^2$ and the six subsets in $\mathcal{H}_J$  delimited by the lines $b_1=q-1$, $b_1= 2q -2$ and $b_2= r(q)+1$.

We will use Proposition \ref{prop3} for obtaining sets of monomials in $\mathcal{H}_J$ whose evaluation give self-orthogonal linear codes. We start by observing that an exponent pair of the form $(x,q-1)$ is not compatible with a pair $(q^2-1 -qx \mod(q^2-1),0)$, where we choose a nonnegative integer less than $q^2$ as a representative of the mentioned congruence class. Then, taking into account that the hyperbola $(b_1+1)(b_2+1) =t$ meets the line $b_2=q-1$ at the point $((t/q)-1,q-1)$; the line $b_2=0$ at the point $(t-1,0)$ and each point on $b_2=q-1$ eliminates a point on $b_2=0$ after moving back $q$ units modulo $q^2 -1$ its non-vanishing coordinate, we get that when $t$ satisfies
\begin{equation}
\label{estrellacirc}
q \left( \frac{t}{q} -1\right) < q^2 -t +1,
\end{equation}
the pairs in $N(J,t)$ on $b_2=q-1$ can be considered as exponents of monomials in the generating set of our supporting linear code. So, from (\ref{estrellacirc}), we can deduce the imposed condition which is $t < (q^2+q+1)/2$.

To guarantee our result, we need to impose some more conditions. We say that two exponents are incompatible if the evaluation the corresponding monomials are not orthogonal. On the one hand, it holds that pairs $(q-1,a)$ on the line $b_1=q-1$ are incompatible with pairs $(q^2-1 -q(q-1),q-1-q a)$, which after taking the corresponding congruence classes provides pairs $(q-1,q-1-a)$ at the same line. This imposes the condition that the intersection point between $(b_1 +1)(b_2+1) =t $ and $b_1= q-1$ should be either $(q-1, r(q) +1)$ or another point on the line $b_1=q-1$ with second coordinate less than $r(q)+1$, note that it is crucial that $t < qr(q)+2q+1$. We must also notice that $r_q(q^2-1)=q-2$ and, by Theorem \ref{FF}, those points in the rectangle $0 \leq b_1 \leq q-2$, $0 \leq b_2 \leq q-2$ are suitable for us.

Now consider the set of points $B$ in $\mathcal{H}_J$ within the rectangle
\[
\left\{ (b_1,b_2) \;|\; q \leq b_1 \leq 2q -3; \; \; 1 \leq b_2  \leq r(q)
\right\}.
\]
The fact that, under the above restriction  for $b_1$, $(q^2-1)-qb_1 = q^2 -2 - q(b_1 -q)$ modulo $q^2-1$ shows that points in $B$ discard points in
\[
\left\{ (b_1,b_2) \;|\; 2q-1 \leq b_1 \leq q^2-1; \; \; r(q)+1 \leq b_2  \leq q-1
\right\}.
\]
Furthermore points in the line $b_1=2q-2$ are incompatible with other points on the same line. As a consequence we add another restriction to avoid most points on this line that can be expressed by saying that the hyperbola $(b_1+1)(b_2 +1)=t$ must intersect the line $b_1=2q-1$  at the point $(2q-1,1)$ or below it, which forces us to consider $t <  4q +1$. In fact, the unique point under the hyperbola that deserves especial attention is $(2q-2,1)$ but it eliminates $(2q-2,q-2)$ and self-orthogonality is preserved. The remaining points on the line $b_1=2q-2$ have the form $(2q-2,b)$, $b \geq 2$ and are not below the hyperbola since $(2q-1)(b+1) \geq 6q-3$, $6q -3 > 4q$ and $t \leq 4q$.

Finally, we are going to prove that
\[
\frac{q^2+q+1}{2} < q r(q) +2q +1.
\]
And, as a consequence, we can eliminate $q r(q) +2q+1$ from Inequality (\ref{de3a2}). Indeed, when $q$ is even,
\[
q r(q) +2q+1 = q \left( \frac{q-2}{2} + 2\right) +1 = \frac{q^2 +2q +2}{2}.
\]
Otherwise
\[
q r(q) +2q+1 = q \left( \frac{q-1}{2} + 1\right) +1 = \frac{q^2 +q +2}{2},
\]
and the proof is completed.
\end{proof}

\begin{exa}
\label{betabeta}
{\rm
Next, we apply Corollary \ref{LLL} to give some examples of good stabilizer codes.

{\it Example A.} Set $q=4$, $N_1= 16$, $N_2=4$ and $J=\emptyset$, then $n_J=64$. Since $(q^2+q+1)/2 = 21/2$ and $4q+1=17$, we can use Corollary \ref{LLL} and get quantum codes up to distance $10$. In fact, it can be checked that self-orthogonality also happens for distances $11$ and $12$, and so, we can derive codes with parameters as in Table \ref{ZZZ}.
\begin{table}
\begin{center}
\begin{tabular}{||c|c|c|c|c|c|c|c|c|c|c|c||}
  \hline
  $k$ & 62 & 58 & 54 & 48 & 46 & 40 &  38 & 32 & 28 & 24& 22 \\
   \hline
 $\geq d$ & 2 & 3 & 4 &  5 & 6 & 7 & 8 & 9 &10 & 11 & 12 \\
  \hline
    & GV & GV & GV * & *& GV *& & &* & &* & * \\
  \hline
\end{tabular}
\end{center}
\caption{Parameters of quantum codes of length 64 over $\mathbb{F}_4$}
\label{ZZZ}
\end{table}
Comparing with \cite[Table III]{lag3} and \cite{ham}, we get four codes with better relative parameters and two new ones, which we label with a star. We also add a symbol GV for identifying those codes exceeding the Gilbert-Varshamov bounds.

{\it Example B.} Assume $q=9$, $N_1=81$, $N_2=9 $ and $J=\emptyset$. Applying  Corollary \ref{LLL} $n_J= 729$ and $t$ should be less than $37$ because $(q^2+q+1)/2 = 91/2$ and $4q+1=37$. Table \ref{tabla3}  shows the parameters of some of our codes.

\begin{table}
\begin{center}
\begin{tabular}{||c|c|c|c|c|c|c|c|c|c||}
\hline
  $k$ & $\geq d$ & $k$ & $\geq d$ & $k$ & $\geq d$ & $k$ & $\geq d$ & $k$ & $\geq d$ \\
  \hline
  727 & 2 & 723 & 3 & 719 & 4 & 713 & 5 & 709 & 6  \\
  \hline
  701 & 7 & 697 & 8 & 689 & 9 & 683 & 10 & 677 & 11\\
 \hline
  675 & 12 & 665 & 13 & 663 & 14 & 657 & 15 & 651 & 16\\
  \hline
  643 & 17 & 641 & 18 & 631 & 19 & 629 & 20 & 621 & 21\\
  \hline
\end{tabular}
\end{center}
\caption{Parameters of quantum codes of length 729 over $\mathbb{F}_9$}
\label{tabla3}
\end{table}
}
\end{exa}

Our previous results consider hyperbolic codes as the supporting linear codes and assume that $p$ divides $N_j$ for $j \not \in J$. Both assumptions can be avoided; in fact one only needs the fact that $p$ divides $N_j$ for at least one index $j$ whenever $J=\emptyset$. A general theorem in this direction would be dense. Therefore, we give the corresponding result in the bivariate case only, which will give good stabilizer codes and shows the underlying ideas. We state it for Euclidean inner product but notice the important fact that the result also holds for Hermitian inner product.
 
\begin{pro}
\label{pnodivide}
Suppose $Q=q$, consider positive integers $N_1$ and $N_2$ such that $N_j-1$ divides $q-1$ for $j=1,2$ and a subset $J \subseteq \{1,2\}$. For $J=\emptyset$, assume that $p$ divides $N_j$ for, at least, one index $j=j_0$. Let $\mathcal{N}$  be a set of monomials in $R$ such that:
\begin{enumerate}[(i)]
\item $\mathcal{N} \subset \mathcal{R}_i(J)$ for either $i=1$ or $i=2$ when $J=\{1,2\}$.
    \item $\mathcal{N} \subset \mathcal{R}_i(J)$ when $J=\{i\}$ and, in case $p \not | ~  N_j$, ($j \neq i$), $\mathcal{N}$ does not contain any pair of monomials of the form $X_i^{b_i} X_j^{0}$ and $X_i^{b'_i} X_j^{N_j -1}$, $j \neq i$, except when  $b_i = b'_i$. That is, we only admit pairs of monomials $X_i^{b_i} X_j^{0}$ and $X_i^{b_i} X_j^{N_j -1}$ and monomials with exponent in $j$, either always $0$ or always $N_j -1$.
    \item $\mathcal{N} \subset R_{j_0}(J)$ when $J=\emptyset$ and we do not consider $X_i^{N_i -1} X_{j_0}^{b_{j_0}}$, $i \neq j_0$, in $\mathcal{N}$ except when $X_i^{0} X_{j_0}^{b_{j_0}}$ is also in $\mathcal{N}$. Moreover, we consider monomials with exponent in $i$, either always $0$ or always $N_j -1$.
\end{enumerate}
 Denote by $\mathcal{N}^l$ the set of exponent pairs of the monomials in $\mathcal{N}$ and write $$\mathcal{N}^{D,l} = \mathcal{H}_J \setminus \{(N_1-1-b_1,N_2-1-b_2)\; | \; (b_1,b_2) \in \mathcal{N}^l\}.$$

Then, the code $E_\mathcal{N} : = E_\mathcal{N}^J$ (see Definition \ref{def:unouno}) is self-orthogonal and the distance of the dual code is larger than or equal to
\[
\delta = \min \left\{ \prod_{j=1}^2 \left(N_j - \epsilon_j -a_j \right) \; | \; \boldsymbol{a} \in \mathcal{N}^{D,l} \right\}.
\]
As a consequence, we are able to construct a stabilizer code with parameters $$\left[\left[n_J, n_J - 2\; \mathrm{card}(\mathcal{N}), \geq \delta\right]\right]_q.$$
\end{pro}
\begin{proof}
Proposition \ref{prop1} shows that the inclusions imposed in the proposition imply the self-orthogonality of the code $E_\mathcal{N}$. Indeed, by Proposition \ref{prop1}, when $p$ divides $N_j$ for all $j \notin J$, it holds that every monomial $X_1^{b_1} X_2^{b_2}$, $b_1 < N_1 -1$, $b_2 < N_2 -1$, admits a unique  monomial with exponents in $\mathcal{H}_J$, $X_1^{N_1 -1-b_1} X_2^{N_2-1-b_2}$, such that the inner product of their evaluations by $\mathrm{ev}_J$ does not vanishes. However when $p$ does not divide $N_j$ for some $j \not \in J$, say $J=\{2\}$ and $p$ does not divide $N_1$, it happens that $\mathrm{ev}_J(X_1^{0} X_2^{b_2}) \cdot \mathrm{ev}_J(X_1^{0} X_2^{N_2-1-b_2}) = \alpha \neq 0$ and $\mathrm{ev}_J(X_1^{0} X_2^{b_2}) \cdot \mathrm{ev}_J(X_1^{N_1-1} X_2^{N_2-1-b_2}) = \beta \neq 0$.  In addition, when one admits $N_j -1$ as an exponent, it also holds that $\mathrm{ev}_J(X_1^{N_1 -1} X_2^{b_2}) \cdot \mathrm{ev}_J(X_1^{0} X_2^{N_2-1-b_2})  \neq 0$ and $\mathrm{ev}_J(X_1^{N_1 -1} X_2^{b_2}) \cdot \mathrm{ev}_J(X_1^{N_1-1} X_2^{N_2-1-b_2})  \neq 0$ and the same result is true if we switch the roles of the indices $1$  and $2$. Notice that the monomial $X_1^{N_1 -1} X_2^{N_2-1}$ is not eligible for our codes. As a consequence, we get that our definition of sets $\mathcal{N}$ avoid the existence of conflicting pairs in $\mathcal{N}^l$, and therefore the evaluation of any monomial $X^{\boldsymbol{b}} \in \mathcal{N}$ is orthogonal to that of  all monomials in $\mathcal{N}$, which proves the self-orthogonality of the code.

To conclude our proof, notice that the dual code $(E_\mathcal{N})^\perp$ could not be only generated by evaluating monomials. This can be checked by computing the dimension $n_J - \dim E_\mathcal{N}$ and the set of monomials whose evaluation is not orthogonal to the evaluation of the monomials generating our code. In fact,  the existence of monomials of the form (say) $X_i^{b_i}X_j^0$ (or $X_i^{b'_i}X_j^{N_j -1}$) forces to eliminate from the dual code the evaluation of $X_i^{N_i -1-b_i}X_j^{N_j -1}$ and $X_i^{N_i -1-b_i}X_j^{0}$ (or $X_i^{N_i -1-b'_i}X_j^{N_j -1}$ and $X_i^{N_i -1-b'_i}X_j^{0}$) because both evaluations are not orthogonal to that of the previous given monomial. Notwithstanding, conditions in (ii) or (iii) allow us to consider pairs in $\mathcal{N}^l$ corresponding to monomials with exponents $0$ and $N_j -1$, when the other exponent is the same, since both monomials eliminate the same pair of monomials.

Furthermore the same conditions show that we also admit monomials with exponents always all $0$ or always $N_j -1$. This is imposed  because we can obtain their corresponding generating polynomials for  the dual code and, as a consequence, we can decide about minimum distance. Let us see it, for instance, for the case $J=\{2\}$ and $p$ does not divide $N_1$. The remaining cases can be reasoned analogously. Notice that, in this case, $(E_\mathcal{N})^\perp$ could not be  generated only by evaluating monomials, however, taking the above given values $\alpha$ and $\beta$, it is clear that
\begin{equation}
\label{binomial}
\mathrm{ev}_J \left(X_1^{0} X_2^{b_2}\right) \cdot \mathrm{ev}_J \left(  \frac{1}{\alpha} X_1^{0} X_2^{N_2-1-b_2} - \frac{1}{\beta} X_1^{N_1 -1} X_2^{N_2-1-b_2} \right) = 0,
\end{equation}
which proves that $(E_\mathcal{N})^\perp$ is generated by evaluating monomials $X_1^{a_1} X_2^{a_2}$ such that
\[
(a_1,a_2) \in \mathcal{H}_J \setminus \left\{(N_1 -1-b_1,N_2 -1-b_2) \; | \; b_1 \neq 0; (b_1,b_2) \in \mathcal{N}^l \right\}
\]
and binomials as  in (\ref{binomial}) for each pair $(0,b_2) \in \mathcal{N}^l$. This concludes our proof by considering a suitable monomial ordering for which $X_1^{0} X_2^{N_2-1-b_2}$ is the leading monomial of $ \frac{1}{\alpha} X_1^{0} X_2^{N_2-1-b_2} - \frac{1}{\beta} X_1^{N_1 -1} X_2^{N_2-1-b_2}$ and after applying Proposition \ref{R}.
\end{proof}

\begin{rem}
{\rm
Proposition \ref{pnodivide} is stated for the case $Q=q$ but, as mentioned, it also holds when $Q=q^2$ and the Hermitian inner product is considered. The difference is that we should consider a set of monomials $\mathcal{N}_q$ such that
\begin{enumerate}[(i)]
\item $\mathcal{N}_q \subset \mathcal{R}_i^q(J)$ for either $i=1$ or $i=2$ when $J=\{1,2\}$.
    \item $\mathcal{N}_q \subset \mathcal{R}_i^q(J)$ when $J=\{i\}$ and, in case $p \not | ~ N_j$, ($j \neq i$), $\mathcal{N}_q$ does not contain any pair of monomials $X_i^{b_i} X_j^{0}$ and $X_i^{b'_{i}} X_j^{N_j-1}$ except when  $b_i = b'_i$; that is, we only admit pairs of monomials of the form $X_i^{b_i} X_j^{0}$ and $X_i^{b_i} X_j^{N_j -1}$ and monomials with exponent in $j$, either always $0$ or always $N_j -1$.
    \item $\mathcal{N}_q \subset R_{j_0}^q(J)$ when $J=\emptyset$ and we do not allow $X_i^{N_i -1} X_{j_0}^{b_{j_0}}$, $i \neq j_0$, in $\mathcal{N}_q$ except when $X_i^{0} X_{j_0}^{b_{j_0}}$ is also in $\mathcal{N}_q$. We also permit monomials with exponents in $i$, either always $0$ or always $N_j -1$.
\end{enumerate}
And, for computing the distance, $\mathcal{N}^{D,l}$ should be replaced with
$$\mathcal{N}_q^{D,l} = \mathcal{H}_J \setminus \{(N_1-1-qb_1 \mod T_1,\;\; N_2-1-qb_2 \mod T_2\;\; | \; \; (b_1,b_2) \in \mathcal{N}_q^l\}.$$  Proposition \ref{prop3}, some ideas in Theorem \ref{FF} and Proposition \ref{R} establish the proof.}
\end{rem}

\subsection{The subfield-subcode case}
We conclude our paper with this subsection, where we show that some similar results can be stated by considering subfield-subcodes. Here we use the notation introduced after Proposition \ref{AA} and in the first two paragraphs of Subsection \ref{hermiticoo}. Ideas behind the proof of our results are supported in previous results.

Next, keeping the above notations, we state our first two results for which we will give a joint proof. With respect to quantum stabilizer codes determined by Euclidean inner product we get the following theorem:
\begin{teo}
\label{AS}
Let $r$ and $s$ be positive integers such that $s$ divides $r$. Set $Q=q=p^r$ and consider values $N_j$, $1 \leq j \leq m$, as above. Consider a subset $J \subseteq \{1,2, \ldots, m\}$ and assume that $p$ divides $N_j$ for all $j \notin J$. Consider also the minimal cyclotomic sets $\mathfrak{I}_{\boldsymbol{a}}$, $\boldsymbol{a} \in \mathcal{A} $ with respect to $p^s$ and, for a positive integer $t \leq n_J$, set
\[
\Delta\left(N(J,t)\right) = \bigcup_{\boldsymbol{a} \in \mathcal{A} \cap N(J,t)} \mathfrak{I}_{\boldsymbol{a}}.
\]
Denote also by $\mathfrak{I}_{-\boldsymbol{a}}$ the cyclotomic set with respect to $p^s$ containing the congruence classes of the tuple $-\boldsymbol{a}$, and assume that $\mathfrak{I}_{\boldsymbol{a}} \neq \mathfrak{I}_{-\boldsymbol{a'}}$ for all pair $\boldsymbol{a}$, $\boldsymbol{a'}$ of elements in $\mathcal{A} \cap N(J,t)$. Then, the subfield-subcode over $\mathbb{F}_{p^s}^{n_J}$ of the $J$-affine variety code given by $\Delta\left(N(J,t)\right)$, $E^{J,\sigma}_{\Delta\left(N(J,t)\right)} = E_{\Delta\left(N(J,t)\right)}^J \cap \mathbb{F}_{p^s}^{n_J}$, satisfies
\[
E^{J,\sigma}_{\Delta\left(N(J,t)\right)} \subseteq \left( E^{J,\sigma}_{\Delta\left(N(J,t)\right)}
\right)^\perp
\]
and gives rise to a stabilizer code with parameters
$[[n_J, n_J - 2 \sum_{\boldsymbol{a} \in \mathcal{A} \cap N(J,t)} i_{\boldsymbol{a}}, \geq t ]]_{p^s}$.
\end{teo}

Our result when one considers the  Hermitian inner product is:

\begin{teo}
\label{CS}
Let $r$ and $s$ be positive integers such that $s$ divides $r$. Set $Q=q^2=p^{2r}$ and consider values $N_j$, $1 \leq j \leq m$, as above (i.e., $N_j -1$ divides $Q-1$). Consider a subset $J \subseteq \{1,2, \ldots, m\}$ and assume that $p$ divides $N_j$ for all $j \notin J$. Consider also the minimal cyclotomic sets $\mathfrak{I}_{\boldsymbol{a}}$ with respect to $p^{2s}$, where $\boldsymbol{a} \in \mathcal{A}$ and for a positive integer $t \leq n_J$, define $
\Delta\left(N(J,t)\right) = \bigcup_{\boldsymbol{a} \in \mathcal{A} \cap N(J,t)} \mathfrak{I}_{\boldsymbol{a}}$. Denote also by $\mathfrak{I}_{-p^s \boldsymbol{a}}$ the cyclotomic set with respect to $p^{2s}$ containing the  congruence classes of the tuple $-p^s \boldsymbol{a}$ and assume that $\mathfrak{I}_{\boldsymbol{a}} \neq \mathfrak{I}_{-p^s \boldsymbol{a'}}$ for all $\boldsymbol{a}$, $\boldsymbol{a'}$ in $\mathcal{A} \cap N(J,t)$. Then the subfield-subcode of the $J$-affine variety code given by $\Delta\left(N(J,t)\right)$ over $\mathbb{F}_{p^{2s}}^{n_J}$, $E^{J,\sigma}_{\Delta\left(N(J,t)\right)} = E_{\Delta\left(N(J,t)\right)}^J \cap \mathbb{F}_{p^{2s}}^{n_J}$, satisfies
\[
E^{J,\sigma}_{\Delta\left(N(J,t)\right)} \subseteq \left( E^{J,\sigma}_{\Delta\left(N(J,t)\right)}
\right)^{\perp_h}
\]
and it gives rise to a stabilizer code with parameters
$[[n_J, n_J - 2 \sum_{\boldsymbol{a} \in \mathcal{A} \cap N(J,t)} i_{\boldsymbol{a}}, \geq t ]]_{p^s}$.
\end{teo}

Proofs for Theorems \ref{AS} and  \ref{CS} are analogous except that, for the first result, we have to consider Euclidean dual and subfield-subcodes over $\mathbb{F}_{p^{s}}^{n_J}$. So, let us show a proof for Theorem  \ref{CS}.

For a start, consider the code, over $\mathbb{F}_{p^{2s}}^{n_J}$, $E^{J,\sigma}_{\Delta\left(N(J,t)\right)}$. By \cite[Theorem 4]{galindo-hernando},  its dimension is $\sum_{\boldsymbol{a} \in \mathcal{A} \cap N(J,t)} i_{\boldsymbol{a}}$. Now a similar reasoning as we did in the proof of Proposition \ref{D} proves that the Hermitian dual of the code $E^{J,\sigma}_{\Delta\left(N(J,t)\right)}$ coincides with the trace $\mathbf{tr}_{2r}^{2s}: \mathbb{F}_{p^{2r}}^{n_J} \rightarrow \mathbb{F}_{p^{2s}}^{n_J}$ of the dual code of $E^{J}_{\Delta\left(N(J,t)\right)}$ with respect to the above defined  inner product $\cdot_{s}$. Notice that $\mathbf{tr}_{2r}^{2s}$ is given componentwise by the map $tr_{2r}^{2s}$ defined in the mentioned proof of Proposition \ref{D}. By the proof of \cite[Theorem 7]{QINP}, this trace is the subfield subcode over $\mathbb{F}_{p^{2s}}^{n_J}$ of the mentioned dual code with respect to $\cdot_s$.

The above ideas allow us to prove the inclusion
\[
E^{J,\sigma}_{\Delta\left(N(J,t)\right)} \subseteq \left( E^{J,\sigma}_{\Delta\left(N(J,t)\right)}
\right)^{\perp_h}.
\]
 
Indeed, by Proposition \ref{prop3} and  \cite[Theorem 3]{galindo-hernando}, $\left( E^{J,\sigma}_{\Delta\left(N(J,t)\right)}\right)^{\perp_h}$ is obtained by evaluating polynomials defined by representatives of minimal cyclotomic sets in
\[
\mathcal{H}_J \setminus \left\{\left([-p^s a_1]_{N_1 -1}, [-p^s a_2]_{N_2 -1}, \ldots, [-p^s a_m]_{N_m -1}\right) \; | \; \boldsymbol{a}=(a_1,a_2, \ldots, a_m) \in \Delta\left(N(J,t)\right) \right\},
\]
therefore, the condition $\mathfrak{I}_{\boldsymbol{a}} \neq \mathfrak{I}_{-p^s \boldsymbol{a'}}$ for all $\boldsymbol{a}$, $\boldsymbol{a'}$ in $\mathcal{A} \cap N(J,t)$, given in the statement, allows us to conclude the mentioned self-orthogonality of the code $E^{J,\sigma}_{\Delta\left(N(J,t)\right)}$.

It only remains to check that the distance is greater than or equal to $t$. In fact, for a linear code $C$ and a positive integer $k$, if we write $C^k$ for the set of words $\mathbf{x}^k =(x_1^k, x_2^k, \ldots, x_m^k)$, then the following inequality of distances of codes holds: $d \left( E^{J,\sigma}_{\Delta\left(N(J,t)\right)}
\right)^{\perp_h} \geq d \left( E^{J}_{\Delta\left(N(J,t)\right)}
\right)^{\perp_{p^s}}$
and, moreover, it also holds that
\[
\left(
E^{J}_{\Delta\left(N(J,t)\right)}
\right)^{\perp_{p^s}} = \left[ \left(
E^{J}_{\Delta\left(N(J,t)\right)}
\right)^{\perp} \right]^{p^{2r-s}},
\]
because $\mathbf{y}^{p^r} = (\mathbf{y}^{p^s})^{p^{2r-s}}$ and $\mathbf{y} \in \left( E^{J}_{\Delta\left(N(J,t)\right)}
\right)^{\perp_{p^s}}$ if and only if $\mathbf{y}^{p^s} \in \left( E^{J}_{\Delta\left(N(J,t)\right)}
\right)^{\perp}$ if and only if $\mathbf{y} \in \left[ \left( E^{J}_{\Delta\left(N(J,t)\right)}
\right)^{\perp} \right]^{p^{2r-s}}$. This concludes the proof because $\left[ \left( E^{J}_{\Delta\left(N(J,t)\right)}
\right)^{\perp} \right]^{p^{2r-s}}$ is isometric to $\left( E^{J}_{\Delta\left(N(J,t)\right)}
\right)^{\perp}$, whose distance is larger than or equal to $t$ since $E_{N(J,t)}^J \subseteq E_{\Delta(N(J,t))}^J$.\\

To finish, we state the following result which follows straightforwardly because our hypotheses on cyclotomic sets in Theorems \ref{AS} and  \ref{CS} can be guaranteed by imposing conditions to one variable, as we showed in Propositions \ref{A} and \ref{D}.

\begin{cor}
\label{el35}
Let $p, R, S, Q$ and $J$ as above. Set $N_1=p^R$ and suppose that $N_j -1$ divides $p^R -1$ for $j \in J\setminus\{1\}$. Denote by $\mathcal{A}^1$ a set of univariate representatives as in Section \ref{seccdos} of minimal cyclotomic sets with respect to $p^s$ modulo $N_1 -1$. Set $
\mathfrak{a} := \max \{a\in\mathcal{A}^1\mid a<p^{R/2}-1\}$ when $R$ is even; and  $\mathfrak{a} := \max \{a\in\mathcal{A}^1 \mid a\le p^{(R+1)/2}-p-1 \}$, otherwise. Then:
\begin{itemize}
\item  In the Euclidean case, that is $Q=q=p^r$, $R=r$ and $S=s$, it holds that $E_{\Delta(N(J,t))}\subset E_{\Delta(N(J,t))}^{\perp}$ for any $t\le \mathfrak{a} + 1$. Therefore, the above self-orthogonal code gives rise to a quantum code with parameters $[[n_J,n_J -2 \, \mathrm{card}(\Delta(N(J,t))),\ge t]]_{p^s}$.
    \item In the Hermitian case, that is $Q=p^{2r}$, $R=2r$ and $S=2s$, it holds that $E_{\Delta(N(J,t))}\subset E_{\Delta(N(J,t))}^{\perp_h}$ for any $t\leq \mathfrak{a} +1$ and the above self-orthogonal code yields a code with parameters $[[n_J,n_J -2 \, \mathrm{card}(\Delta(N(J,t))),\geq t]]_{p^s}$.
    \end{itemize}
\end{cor}

    In addition, when $J \subseteq \{1,2\}$ we can reason as in corollaries \ref{L} and \ref{LL} and, then, the above results are true for $t\le \mathfrak{a} + 2$ whenever $J=\emptyset$ or $J=\{2\}$.

\begin{exa}
\label{gamagama}
{\rm
Above results are the support of the following examples of quantum codes.

{\it Example A.}
With notation as in Theorem \ref{CS}, set $p=5$, $r=2$, $s=1$, $Q= 5^4$, $J=\emptyset$, $N_1=15$ and $N_2=6$. Consider $t=3$ and then $
\Delta\left(N(J,t)\right)  = \{ (0, 0),  (1, 0),  (12, 0),  (0, 1)\}$ which satisfies the hypotheses in the mentioned theorem. Then, we get a $[[70,62, \geq 3]]_5$ quantum code which can be extended to a quantum code with parameters $[[71,62, \geq 3]]_5$, which improves the $[[71,61, \geq 3]]_5$ quantum code given in \cite[Table 1]{lag2} and  \cite{edel}.

{\it Example B.} We conclude with an example obtained applying Corollary \ref{el35}. Set $p=2$, $r=4$, $s=2$, $N_1=256$, $N_2=2$ and $J=\emptyset$. Here $\mathfrak{a}=14$ and we get codes of length $512$ over $\mathbb{F}_4$. Table \ref{TT} displays the corresponding parameters. Note that the codes with distance less than 9, except that with distance 7, exceed the Gilbert-Varshamov bound.

\begin{table}
\begin{center}
\begin{tabular}{||c|c|c|c|c|c|c|c|c|c|c|c|c|c|c|c||}
  \hline
  $k$ & 510 & 504 & 500 & 492 & 488 & 480 &  476 & 468 & 464 & 456& 452& 444&  440 & 432 &  428 \\
   \hline
 $\geq d$ & 2 & 3 & 4 &  5 & 6 & 7 & 8 & 9 & 10 & 11 & 12& 13& 14& 15&  16 \\
  \hline
\end{tabular}
\end{center}
\caption{Parameters of quantum codes of length $512$ over $\mathbb{F}_4$}
\label{TT}
\end{table}

}
\end{exa}


\begin{thebibliography}{99}
\footnotesize \setlength{\baselineskip}{3mm}

\bibitem{Akk} Aly, S.A., Klappenecker, S., Sarvepalli, P.K. On quantum and classical BCH codes, {\it IEEE Trans. Inf. Theory} {\bf 53} (2007) 1183-1188.

\bibitem{geil} Andersen, H.E., Geil, O. Evaluation codes from order domain theory, {\it Finite Fields Appl.} {\bf 14} (2008) 92-123.

\bibitem{7kkk} Ashikhmin, A., Barg, A., Knill, E., Litsyn, S. Quantum error-detection I: Statement of the problem, {\it IEEE Trans. Inf. Theory} {\bf 46} (2000) 778-788.

\bibitem{8kkk} Ashikhmin, A., Barg, A., Knill, E., Litsyn, S. Quantum error-detection II: Bounds, {\it IEEE Trans. Inf. Theory} {\bf 46} (2000) 789-800.

\bibitem{AK} Ashikhmin, A., Knill, E. Non-binary quantum stabilizer codes, {\it IEEE Trans. Inf. Theory} {\bf 47} (2001) 3065-3072.

 

\bibitem{BE} Bierbrauer, J., Edel, Y. Quantum twisted codes, {\it J. Comb. Designs} {\bf 8} (2000) 174-188.


\bibitem{18kkk} Calderbank, A.R., Rains, E.M., Shor, P.W., Sloane, N.J.A. Quantum error correction and orthogonal geometry, {\it Phys. Rev. Lett.} {\bf 76} (1997) 405-409.

\bibitem{19kkk} Calderbank, A.R., Rains, E.M., Shor, P.W., Sloane, N.J.A. Quantum error correction via codes over GF(4), {\it IEEE Trans. Inf. Theory} {\bf 44} (1998) 1369-1387.

\bibitem{20kkk} Calderbank A.R., Shor, P. Good quantum error-correcting codes
exist, {\it Phys. Rev. A} {\bf 54} (1996) 1098-1105.

\bibitem{chen} Chen, B., Ling, S. Zhang, G. Application of constancyclic codes to quantum MDS codes, {\it IEEE Trans. Inform. Theory} {\bf 61} (2015) 1474-1484.

 

\bibitem{CLO} Cox, D., Little, J., O'Shea, D. \newblock {\em Ideals, Varieties, and Algorithms: an Introduction to Computational Algebraic Geometry and Commutative Algebra}. \newblock Undergraduate Texts in Mathematics, forth edition. Springer-Verlag, 2015.

\bibitem{8AS} Dieks, D. Communication by EPR devices, {\it Phys. Rev. A} {\bf 92} (1982) 271.

\bibitem{edel} Edel, Y. \emph{Some good quantum twisted codes}.
Online available at {\tt http://\-www.\-mathi.\-uni-hei\-delberg.de\-/~yves/Matritzen/QTBCH/QTBCHIndex.html}.


\bibitem{eck} Ekert, A., Macchiavello, C. Quantum error correction for communication, {\it Phys. Rev. Lett.} {\bf 77} (1996) 2585.

 

\bibitem{35kkk} Feng, K. Quantum error correcting codes. In Coding Theory and Cryptology, Word Scientific, 2002, 91-142.

\bibitem{feng} Feng, K., Ma, Z.  A finite Gilbert-Varshamov bound for pure stabilizer quantum codes, {\it IEEE Trans. Inf. Theory} {\bf 50} (2004) 3323-3325.
 

\bibitem{galindo-hernando} Galindo, C., Hernando, F. Quantum codes from affine variety codes and their subfield subcodes, {\it Des. Codes Crytogr.} {\bf 76} (2015) 89-100.

\bibitem{gal-her-rua} Galindo, C., Hernando, F., Ruano, D. New quantum codes from evaluation and matrix-product codes, {\it Finite Fields Appl.} {\bf 36} (2015) 98-120.

\bibitem{QINP} Galindo, C., Hernando, F., Ruano, D. Stabilizer quantum codes from $J$-affine variety codes and a new Steane-like enlargement, {\it Quantum Inf. Process.} {\bf 14} (2015) 3211-3231.

\bibitem{rey}  Galindo, C., P\'erez-Casales, R. On the evaluation codes given by simple $\delta$-sequences, {\it AAECC} {\bf 27} (2016), 59-90.

\bibitem{galmon} Galindo C., Monserrat, F. Delta-sequences and evaluation codes defined by plane valuations at infinity, {\it Proc. London Math. Soc.} {\bf 98} (2009) 714-740.

\bibitem{galmon2} Galindo C., Monserrat, F. Evaluation codes defined by finite families of plane valuations at infinity. {\it Des. Codes Crytogr.} {\bf 70} (2014) 189-213.

\bibitem{Geil-Affine} Geil, O. {\it Evaluation codes from an affine variety code perspective}.  Advances in algebraic geometry codes, Ser. Coding Theory Cryptol. 5 (2008) 153-180. World Sci. Publ., Hackensack, NJ. Eds.: E. Martinez-Moro, C. Munuera, D. Ruano.



\bibitem{geil3} Geil, O. Roots and coefficients of multivariate polynomials over finite fields, {\it Finite Fields Appl.} {\bf 23} (2015) 35-52.

\bibitem{gh} Geil, O., H{\o}holdt, T. On hyperbolic codes, {\it Lect. Notes Comp. Sc.} {\bf 2227} (2001) 159-171.

\bibitem{geil2} Geil, O., Matsumoto, R., Ruano, D. Feng-Rao decoding of primary codes, {\it Finite Fields Appl.} {\bf 34} (2013) 36-44.

\bibitem{geilthomsen} Geil, O., Thomsen, C. Weighted {R}eed-{M}uller codes revisited, {\it Des. Codes Cryptogr.} {\bf 66} (2013) 195--220.

 

\bibitem{38kkk} Gottesman, D. A class of quantum error-correcting codes saturating
the quantum Hamming bound, {\it Phys. Rev. A} {\bf 54} (1996) 1862-1868.

\bibitem{codet} Grassl, M. \emph{Bounds on the minimum distance of linear codes}. Online available at {\tt http://www.codetables.de}, accessed on 15th February 2015.

\bibitem{opt} Grassl, M., Beth, T., R\"{o}tteler, M. On optimal quantum codes, {\it Int. J. Quantum Inform.} {\bf 2} (2004) 757-775.

\bibitem{45kkk}  Grassl, M., R\"{o}tteler, M. Quantum BCH codes. In Proc. X Int. Symp. Theor.  elec. Eng. Germany 1999, 207-212.

 

\bibitem{ham} Hamada, M. Concatenated quantum codes constructible in polynomial time: Efficient decoding and error correction, {\it IEEE Trans. Inform. Theory} {\bf 54} (2008) 5689-5704.

\bibitem{refer} He, X., Xu. L., Chen, H. New $q$-ary quantum MDS codes with distances bigger than $q/2$. {\it Quantum Inf. Process.} {\bf 15} (2016) 2745-2758.


 

\bibitem{ho} H{\o}holdt, T. On (or in) Dick Blahut's footprint. {\it Codes, curves and signals} (1998) 3-9.

\bibitem{jin} Jin, L., Ling, S., Luo, J., Xing, C. Application of classical Hermitian self-orthogonal MDS codes to quantum MDS codes, {\it IEEE Trans. Inform. Theory} {\bf 56} (2010) 4735-4740.
 

\bibitem{kkk} Ketkar, A.,  Klappenecker, A., Kumar, S., Sarvepalli, P.K. Nonbinary stabilizer codes over finite fields, {\it IEEE Trans. Inform. Theory} {\bf 52} (2006) 4892-4914.

 

\bibitem{lag3} La Guardia, G.G. Construction of new families of nonbinary quantum BCH codes, {\it Phys. Rev. A} {\bf 80} (2009) 042331.

\bibitem{lag2} La Guardia, G.G. On the construction of nonbinary quantum BCH codes, {\it IEEE Trans. Inform. Theory} {\bf 60} (2014) 1528-1535.

\bibitem{lag1} La Guardia, G.G., Palazzo, R. Constructions of new families of nonbinary CSS codes, {\it Discrete Math.} {\bf 310} (2010) 2935-2945.

 

\bibitem{mas} Massey, J.L., Costello, D.J., Justensen, J. Polynomial weights and code constructions. {\it IEEE Trans. Inform. Theory} {\bf 19} (1973) 101-110.

\bibitem{71kkk} Matsumoto, R., Uyematsu, T. Constructing quantum error correcting codes for $p^m$ state systems from classical error correcting codes. {\it IEICE Trans. Fund.} {\bf E83-A} (2000) 1878-1883.

 

\bibitem{mat} Matsumoto, R., Uyematsu, T. Lower bound for the quantum capacity of a discrete memoryless quantum channel, {\it J. Math. Phys} {\bf 43} (2002) 4391-4403.

 

\bibitem{saint} Saints, K., Heegard, C. On hyperbolic cascaded Reed-Solomon codes, {\it Lect. Notes Comp. Sc.} {\bf 673} (1993) 291-303.

\bibitem{Sarvepalli} Sarvepalli, P.K., Klappenecker, A.
 Nonbinary quantum Reed-Muller codes. In Proc. 2005 Int. Symp. Information Theory,  1023-1027.

 

\bibitem{22RBC} Shor, P.W. Polynomial-time algorithms for prime factorization and discrete logarithms on a quantum computer, in Proc. 35th  ann. symp. found. comp. sc., {\it IEEE Comp. Soc. Press} 1994, 124-134.

\bibitem{23RBC} Shor, P.W. Scheme for reducing decoherence in quantum computer memory, {\it Phys. Rev. A} {\bf 52} (1995) 2493-2496.

 

\bibitem{95kkk} Steane, A.M. Simple quantum error correcting codes, {\it Phys. Rev. Lett.} {\bf 77} (1996) 793-797.

 

\bibitem{Steane-E} Steane, A.M. Enlargement of Calderbank-Shor-Steane quantum codes, {\it IEEE Trans. Inform. Theory } {\bf 45} (1999) 2492-2495.

\bibitem{26RBC} Wootters W.K., Zurek, W.H. A single quantum cannot be cloned, {\it Nature}
{\bf 299} (1982) 802-803.

\bibitem{f2} Yu, S. Bierbrauer, J., Dong, Y., Chen, Q., Oh, C.H. All the stabilizer codes of distance 3, {\it IEEE Trans. Inform. Theory } {\bf 59} (2013) 5179-5185.

\end{thebibliography}
\end{document}